\documentclass[final]{siamltex1213}

\usepackage{amssymb}
\usepackage{amsmath}
\usepackage{graphicx}
\usepackage{esint}

\def\ut{{\tilde{u}}}
\def\Au{{\mathcal{A}_{u}}}
\def\Auk{{\mathcal{A}_{u_k}}}
\def\Aut{{\mathcal{A}_{\tilde{u}}}}
\def\C{{\mathcal{C}}}
\def\Q{{\mathcal{Q}}}
\def\Nalpha{{\mathcal{N}_{\alpha}}}

\def\Nalphat{{\mathcal{N}_{\tilde{\alpha}}}}
\def\Idelta{{\mathcal{I}_{\delta}}}
\def\Hdelta{{H_{\delta}}}

\title{Error bounds for gradient density estimation computed from a finite sample set using the method of stationary phase}
\author{
Karthik S. Gurumoorthy\footnotemark[3]\ \footnotemark[2]
\and Anand Rangarajan\footnotemark[4]\ \footnotemark[5]
}

\begin{document}
\maketitle 
\slugger{sinum}{xxxx}{xx}{x}{x--x}

\renewcommand{\thefootnote}{\fnsymbol{footnote}}
\footnotetext[3]{Corresponding author. International Center for Theoretical Sciences, Tata
 Institute of Fundamental Research, Bangalore, Karnataka, India (\email{karthik.gurumoorthy@icts.res.in}).}
 
 \footnotetext[4]{Department of Computer and Information Science and Engineering, 
  University of Florida, Gainesville, Florida, USA (\email{anand@cise.ufl.edu}).}
  
 \footnotetext[2]{This work benefited from the support of the
  AIRBUS Group Corporate Foundation Chair in Mathematics of Complex
  Systems established in ICTS-TIFR.}
   
\footnotetext[5]{This work was partially funded by  NSF IIS 1065081. }

\begin{abstract}
For a twice continuously differentiable function $S$, we define the density function of its gradient (derivative in one dimension) $s = S^{\prime}$ as a random variable transformation of a uniformly distributed random variable using $s$ as the transformation function. Given $N$ values of $S$ sampled at equally spaced locations, we demonstrate using the method of stationary phase that the approximation error between the integral of the scaled, discrete power spectrum of the wave function $\phi^{D}_{\tau}=\frac{1}{\sqrt{L}}\exp\left(\frac{iS}{\tau}\right)$ and the integral of the true density function of $s$ over an arbitrarily small interval is bounded above by $O(1/N)$ as $N \rightarrow \infty$ ($\tau \rightarrow 0$). In addition to its easy implementation and fast computability in $O(N \log N)$ that only requires computing the discrete Fourier transform, our framework for obtaining the derivative density does not involve any parameter selection like the number of histogram bins, width of the histogram bins, width of the kernel parameter, number of mixture components etc. as required by other widely applied methods like histograms and Parzen windows. 
\end{abstract}


\pagestyle{myheadings} \thispagestyle{plain}
\markboth{KARTHIK GURUMOORTHY et. al.}{ERROR BOUNDS FOR GRADIENT DENSITY ESTIMATION}

\section{Introduction}
Density estimation methods attempt to estimate an unobservable probability density function using observed data \cite{Parzen62,Rosenblatt56,Silverman86,Bishop06}. The observed data are treated as random samples from a large population which is assumed to be distributed according to the underlying density function.  The aim of our current work is to compute the density function of the gradient---corresponding to derivative in one dimension---of a function $S$ (density function of $S^{\prime}$) from finite set of $N$ samples of $S$ using the method of stationary phase \cite{Cooke82,Jones58,McClure91,McClure97,OlverBook74,Wong81} and bound the error between the estimated and the unknown true density as a function of $N$.  If $s = S^{\prime}$ represent the derivative of the function $S$, the density function of $s$ is defined via a random variable transformation of the uniformly distributed random variable $X$ using $s$ as the transformation function. In other words, if we define a random variable $Y=s(X)$ where the random variable $X$ has a uniform distribution on the interval $\Omega = [0,L]$, the density function of $Y$ represents the density function of $s$. 

In the field of computer vision many applications arise where the density of the gradient of the image, also popularly known as the histogram of oriented gradients (HOG), directly estimated from samples of the image are employed for human and object detection \cite{Dalal05,Zhu06}. Here the image intensity plays the role of the function $S$ and the distribution of intensity gradients or edge directions are used as the feature descriptors to characterize the object appearance and shape within an image. In the recent article \cite{Hu13}, an adaption of the HOG descriptor called the Gradient Field HOF (GF-HOG) is used for sketch based image retrieval. 

Our current work is along the lines of our earlier efforts \cite{Gurumoorthy12,Gurumoorthy13}. In \cite{Gurumoorthy12} we focused on exploiting the stationary phase tool to obtain gradient densities of Euclidean distance functions in two dimensions. As the gradient norm of Euclidean distance functions is identically equal to $1$ everywhere, the density of the gradients is one-dimensional and defined over the space of orientations. In \cite{Gurumoorthy13} we generalized and established this equivalence between the power spectrum and the gradient density to arbitrary smooth functions in arbitrary finite dimensions. The fundamental point of departure between our current work and the results proved in \cite{Gurumoorthy12,Gurumoorthy13} is that here we compute the derivative density from a \emph{finite, discrete} samples of $S$, rather than requesting the availability of complete description of $S$ on $\Omega$ as sought in \cite{Gurumoorthy12,Gurumoorthy13}. Given only $N$ samples of $S$ the convergence proof involving continuous Fourier transform in \cite{Gurumoorthy12,Gurumoorthy13} has to be substituted with its discrete counterpart. Aliasing errors \cite{Bracewell99} which are non-existent in the continuous case have to be explicitly addressed in the present discrete setting. Curious enough we find that the free parameter $\tau$, which could be set arbitrarily close to zero in the continuous case, has to respect a lower bound proportional to $1/N$ in the discrete scenario and an apposite value of $\tau$ as a function of $N$ can be explicitly determined. Apart from establishing the equivalence between the power spectrum and the gradient density, we also quantify the approximation error as a function of $N$, a result \emph{not} discussed in \cite{Gurumoorthy12,Gurumoorthy13}. Even in one dimension we find the discrete setting to be challenging and worthy of a separate examination. The discrete, one dimensional case seems to posses most of the mathematical complexities of its higher dimensional counterpart (thought at this point we are not entirely sure) and lays the foundation for extending our bounds on approximation error to arbitrary finite dimensions, a task we plan to take up in the future. 

\subsection{Main Contribution}
Say we have $N$ samples of a function $S$ obtained at uniform intervals of $\delta = \frac{L}{N}$ between $[0,L]$ at locations $y_n = (n+\frac{1}{2}) \delta, 0 \leq n \leq N-1$ denoted by the  set $\{S(y_n)\}_{n=0}^{N-1}$. As before, let $s = S^{\prime}$ denote the derivative of $s$. For all positive parameter $\tau >0$, we define a function
\begin{equation}
\label{def:phidiscrete}
\phi_{\tau}^{D}(y_n) \equiv \frac{1}{\sqrt{L}} \exp\left(\frac{i S(y_n)}{\tau}\right) 
\end{equation}
at these $N$ discrete locations $\{y_n\}_{n=0}^{N-1}$ by expressing $S$ as the phase of the wave function $\phi^D$ and consider its discrete power spectrum at the suitable choice of $\tau \propto \frac{1}{N}$. We show that the approximation error between the integral of this discrete power spectrum over an arbitrary small interval $\Nalpha$ with the interval length chosen \emph{independently} of $N$ and the cumulative measure of the true density of $s$ over $\Nalpha$ is bounded above by $O(1/N)$. The formal mathematical statement of our result is stated in Theorem~\ref{thm:derivativedensity}. In our current effort we affirmatively answer the following questions:
\begin{enumerate}
\item As the number of samples $N \rightarrow \infty$, does the discrete power spectrum (its interval measure to be precise) increasingly closely approximates the true density of the derivatives? 
\item If yes, can we estimate the approximate error as a function of $N$?
\item Is there a lower bound on $\tau$ as a function of $N$ that precludes it from being set arbitrarily close to zero?
\item Is there an optimum value for $\tau$ as a function of $N$?
\end{enumerate}
 We call our approach the \emph{wave function method} for computing the probability density function and henceforth will refer to it as such.
 
 \subsection{Brief exposition of our previous continuous case result}
\noindent The crux of our continuous case results in \cite{Gurumoorthy12,Gurumoorthy13}---when restricted to one dimension---is the fact the frequency values $u$ are the \emph{gradient histogram bins} for the stationary points of the function
\begin{equation}
\label{def:Txu}
T(x;u) = S(x) - ux.
\end{equation}
To elaborate, consider the definition of the continuous case scaled Fourier transform in~(\ref{def:scaledFT}).
The first exponential $\exp\left(\frac{i S(x)}{\tau}\right)$ is a varying
complex {}``sinusoid\textquotedblright{}, whereas the second exponential
$\exp\left(-\frac{iu x}{\tau}\right)$ is a fixed complex sinusoid
at frequency $\frac{u}{\tau}$. When we multiply these two complex
exponentials, at low values of $\tau$, the two sinusoids are usually
not {}``in sync\textquotedblright{} and tend to cancel each other
out. However, around the locations where $s(x)=u$, the two
sinusoids are in perfect sync (as the combined exponent is \emph{stationary})
with the approximate duration of this resonance depending on $S^{\prime\prime}(x)$.
The value of the integral in~(\ref{def:scaledFT}) can be approximated
via the stationary phase approximation \cite{OlverBook74}
as 
\begin{equation*}
F_{\tau}(u)\approx\frac{1}{\sqrt{L}}\exp\left(\pm\frac{i\pi}{4}\right)\sum_{m=1}^{M(u)}\exp\left(\frac{i}{\tau}\left[S(x_{m})-u x_{m}\right]\right) \frac{1}{\sqrt{S^{\prime\prime}(x_m)}}
\end{equation*}
where $M(u)=\left|\mathcal{A}_{u}\right|$, the latter defined in~(\ref{eq:setA}) and $\{x_m\}_{m=1}^{M(u)}$ are the stationary points for the given frequency $u$. The approximation is increasingly
tight as $\tau\rightarrow0$. The power spectrum $P_{\tau}(u)$
gives us the required result $\frac{1}{L}\sum_{k=1}^{M(u)}\frac{1}{\left|S^{\prime\prime}(x_m)\right|}$
except for the \emph{cross} phase factors $S(x_m)-S(x_t)-u(x_m-x_t)$
obtained as a byproduct of two or more remote locations $x_m$ and
$x_t$ indexing into the same frequency bin $u$, i.e, $x_m\not=x_t$,
but $s(x_m)=s(x_t)=u$. Integrating the power spectrum over a small frequency range $\Nalpha$ removes
these cross phase factors and we obtain the intended result.
 
\subsection{Significance of our current result}
The benefits of our wave function method for computing the density function of the derivative are multi fold and are stated below:
\begin{enumerate}
\item One of the foremost advantages of our approach is that it recovers the derivative density function of $S$ \emph{without} explicitly determining
its derivative $s$.  Since the stationary points capture derivative information and map them into
the corresponding frequency bins, we can directly work with $S$ circumventing the need to compute its derivative. 
\item Our method is extremely fast in terms of its computational complexity. Given the $N$ sampled values, the discrete Fourier transform of $\exp\left(\frac{iS(x)}{\tau}\right)$ at the apt value of $\tau$ can be computed in $O(N \log N)$ \cite{Cooley65} and the subsequent squaring operation to obtain the power spectrum can be performed in $O(N)$. Hence the overall time complexity to obtain the density is only $O(N \log N)$.
\item As established subsequently, our wave function method approximates as $O(1/N)$ to the true density. For histograms and the kernel density estimators \cite{Parzen62,Rosenblatt56} the approximation errors are established for the integrated mean squared error (IMSE) expressed as:
\begin{align}
\label{eq:IMSE}
IMSE(N) \equiv \mathbb{E}\|P-\tilde{P}_{N}\|_2^2 = \mathbb{E} \int \left(P(u)-\tilde{P}_{N}(u)\right)^2 \, du,
\end{align}
where $P(u)$ is the true density, $\tilde{P}_N(u)$ is the computed density from the given $N$ samples, and the expectation $\mathbb{E}$ is with respect to samples of size $N$. The approximation error of IMSE for histograms and kernel density estimators are proven to be $O\left(N^{-\frac{2}{3}}\right)$ \cite{Scott79,Cencov62} and $O\left(N^{-\frac{4}{5}}\right)$ \cite{Wahba75} respectively. Prior to asserting that the $O(1/N)$ approximation of our wave function method is superior compared to histograms and kernel density estimators, we would like to caution the reader of the following:
\begin{enumerate}
\item As the sample locations $\{y_n\}_{n=1}^N$ are fixed, taking expectations over all possible samples of size $N$ in order to compute the IMSE loses its relevance in our setting. 
\item The \emph{point-wise} convergence result stated in Theorem~\ref{thm:derivativedensity} is not entirely commensurate with the convergence of IMSE which involves computing the $\ell_2$ error---obtained by integrating the \emph{square} of the difference between the true and the computed probability densities over $all$ the locations---as expressed in (\ref{eq:IMSE}).
\end{enumerate}
Bearing in mind the aforesaid key differences we refrain from drawing any affirmative conclusions.
\item Our framework for obtaining the density does not involve any parameter selection like number of histogram bins, width of the histogram bins, width of the kernel parameter, number of mixture components etc. as required by other widely applied methods like the histograms and the kernel density estimators \cite{Parzen62,Rosenblatt56}. It is worth emphasizing that though $\tau$ appears to be a free parameter in our setting, we explicitly provide an optimal value for $\tau$ computed solely based on the samples of the function $S$.
\end{enumerate}

%

Table~\ref{table:importantsymbols} lists the important symbols used in this article and their interpretations. 

\begin{center}
\begin{table}[ht!]
\caption{List of important symbols \label{table:importantsymbols}}
\noindent \centering{}%
\begin{tabular}{|l|l|}
\hline 
Symbols  & Interpretation \tabularnewline
\hline\hline
$i, \tau$ & The imaginary unit satisfying $i^2 = -1$ and a free parameter respectively. \tabularnewline
\hline
$S,s, \phi, B$ & The true function, its derivative, the sinusoidal function containing $S$  \tabularnewline
& in its phase, and the bound on the derivative respectively. \tabularnewline
\hline
$\delta, N, \Omega, L$ & The sampling interval, number of samples, domain of $S$, and the length  \tabularnewline
& of domain respectively. \tabularnewline
\hline
$\{y_n\}_{n=0}^{N-1}, \{u_k\}_{k=0}^{N-1}$  & Sampling and the frequency locations respectively. \tabularnewline
\hline 
$\{x_m\}_{m=1}^{M(u)}, \{x_t\}_{t=1}^{M(u)}$ & Interchangeable notations of the same finite set $\Au$ of cardinality $M(u)$  \tabularnewline
& containing the stationary points for a given frequency value $u$. \tabularnewline
\hline
$F_{\tau}^D(u_k), P_{\tau}^D(u_k)$ & Scaled discrete Fourier transform and its magnitude square respectively.\tabularnewline
\hline
 $P_{\tau}^{DTFT}(u)$ & Magnitude square of the scaled discrete time Fourier transform. \tabularnewline
\hline
$P(u)$ & The true density of $s$ obtained via random variable transformation. \tabularnewline
\hline
$\epsilon_{r,\tau}, r \in \{1,2,3,4,5\}$ & Error terms. \tabularnewline
\hline 
$f_{\tau}(u) = O(\tau)$ & There exist a constant $\lambda >0$ and a bounded continuous function $\gamma(u)$ \tabularnewline
& both independent of $\tau$ such that when $\tau \leq \lambda$, $|f_{\tau}(u)| \leq \tau \gamma(u)$. \tabularnewline
\hline 
\end{tabular}
\end{table}
\par\end{center}

\section{Nature of the true function $S$ and existence of density function}
\label{sec:Assumption}
The function $S$ is assumed to be twice continuously differentiable, defined on the closed interval $\Omega = [0, L]$ with length $L$ and has a non-vanishing second derivative
\emph{almost everywhere} on $\Omega$, i.e. 
\begin{equation}
\mu\left(\{x \in \Omega :S^{\prime\prime}(x)=0\}\right)=0,
\label{eq:assumption1}
\end{equation} 
where $\mu$ denotes the Lebesgue measure. As clarified below, the assumption in (\ref{eq:assumption1}) is made in order to ensure that the density function of $s$ exists almost everywhere. The required smoothness on $S$ will become clear in the proof of the subsequent lemmas.

\noindent Define the following sets: 
\begin{align}
\label{eq:setB}
\mathcal{B} & \equiv  \{x:S^{\prime\prime}(x)=s^{\prime}(x)=0\}\, \\
\label{eq:setC}
\C & \equiv  \{s(x):x\in\mathcal{B}\}\cup\{s(0),s(L)\},\,\mathrm{and}\, \\
\label{eq:setA}
\Au & \equiv  \{x:s(x)=u\}.
\end{align}
Here, $s(0)=\lim_{x\rightarrow 0^{+}}s(x)$ and $s(L)=\lim_{x\rightarrow (L)^{-}}s(x)$. The higher derivatives of $S$ at the end points $0$ and $L$ are
also defined along similar lines using one-sided limits. The main purpose of
defining these one-sided limits is to exactly determine the set $\mathcal{C}$  where the density of $Y$ is \emph{not defined}. Let $M(u) = |\Au|$. We now state some useful lemmas whose proofs are provided in Appendices~\ref{sec:proofoffinitenesslemma}, \ref{sec:proofofneighborhoodlemma} and \ref{sec:proofofdenistylemma} respectively.

\begin{lemma}{[}Finiteness Lemma{]} 
\label{lemma:finitenessLemma} $\Au$ is finite for every $u\notin\C$.
\end{lemma}

As we see from Lemma~\ref{lemma:finitenessLemma} above that for a given $u \notin \C$, there is only a \emph{finite} collection of $x \in \Omega$ that maps to $u$ under the function $s$. The inverse map $s^{(-1)}(u)$ which identifies the set of $x \in \Omega$ that maps to $u$ under $s$ is ill-defined as a function as it is a one to many mapping. The objective of the following Lemma~\ref{lemma:neighborhoodLemma} is to define appropriate neighborhoods such that the inverse function $s^{(-1)}$ when restricted to those neighborhoods is well-defined.

\begin{lemma}{[}Neighborhood Lemma{]}
\label{lemma:neighborhoodLemma} For every
$u \notin \C$, there exist a closed neighborhood $\Nalpha(u)$ around $u$ such that $\Nalpha(u) \bigcap\C$ is empty. Furthermore, if $M(u)=|\Au| >0$, $\Nalpha(u)$ can be chosen such that we can find a closed neighborhood $\Nalpha(x)$ around each $x \in \Au$ satisfying the following conditions:
\begin{enumerate}
\item $s\left(\Nalpha(x)\right) = \Nalpha(u)$.
\item $S^{\prime \prime}(y) \not=0, \forall y \in \Nalpha(x)$.
\item The inverse function $s^{(-1)}(u):\Nalpha(u) \rightarrow \Nalpha(x)$ is well-defined.
\item $S^{\prime \prime}$ is of constant sign in $\Nalpha(x)$.
\item $M(u)$ is constant in $\Nalpha(u)$.\\
\end{enumerate}
\end{lemma}

\begin{lemma}{[}Density Lemma{]} \label{lemma:densityLemma} The probability
density of $Y$ on $\mathbb{R}-\mathcal{C}$ exists and is given by
\begin{equation}
P(u)=\frac{1}{L}\sum_{m=1}^{M(u)}\frac{1}{\left|S^{\prime\prime}(x_m)\right|},
\label{eq:graddensity}
\end{equation}
where the summation is over $\mathcal{A}_{u}$ (which is the finite
set of locations $x_m \in\Omega$ where $s(x_m)=u$
as per Lemma~\ref{lemma:finitenessLemma}).
\end{lemma} 

From Lemma~\ref{lemma:densityLemma} it is clear that the existence of the density function $P$ at a location $u \in \mathbb{R}$ necessitates that $S^{\prime \prime}(x) \not=0, \forall x \in \Au$. Since we are interested in the case where the density exists almost everywhere on $\mathbb{R}$, we impose the constraint that the set $\mathcal{B}$ in (\ref{eq:setB}), comprising of all points where $S^{\prime \prime}$ vanishes has a Lebesgue measure zero. It follows that $\mu(\mathcal{C})=0$. Though we have a closed form expression for $P(u)$ in (\ref{eq:graddensity}) it is generally hard to compute it directly as for choice of $u$ we need to laboriously determine the set $\mathcal{A}_{u}$. Our wave function approach totally circumvents this difficulty.

\section{Fourier Transform and its discrete version}
Below we define the various versions of the Fourier transforms used in this article. Though the definitions stated here might appear to be slightly different from the standard textbook definitions, it is fairly straightforward to establish their equivalence. We find it to be imperative to explicitly state them in this modified form as they aid in comprehending our results better.
\subsection{Discrete Fourier transform (DFT) and the scaled DFT}
The DFT of the function $\phi_{\tau}(y_n) $ is defined as
\begin{equation}
\label{def:DFT}
F^{D}(w_k) \equiv \delta \sum_{n=0}^{N-1} \phi_{\tau}^{D}(y_n)  \exp(-i 2 \pi w_k y_n)
\end{equation} 
where for $0 \leq k \leq N-1$,
\begin{equation}
\label{eq:freqwk}
w_k = \left\{ \begin{array}{ll}
              \frac{1}{N\delta}\left(k - \frac{N}{2}\right) & \mbox{if $N\%2=0$}; \\ & \\
              \frac{1}{N\delta}\left(k - \frac{N-1}{2}\right) & \mbox{if $N\%2=1$}.
          \end{array} \right.
\end{equation}
The shifts introduced in the definition of $w_k$ places the zero-frequency component at the center of the spectrum. The inverse DFT is given by
\begin{equation*}
\phi_{\tau}^{D}(y_n)  = \frac{1}{N \delta}\sum_{k=0}^{N-1} F^{D}(w_k)\exp(i 2 \pi w_k y_n).
\end{equation*}
For the subsequent analysis we assume that $N$ is even. It is worth emphasizing that $\frac{1}{N \delta}$ is the interval between the frequencies where the DFT values are defined. Traditionally, the definition of DFT and its inverse doesn\rq{}t explicitly include the sampling interval $\delta$ as it is generally unknown and set to 1. 

Let $u_k = 2 \pi \tau w_k, 0\leq k \leq N-1$ denote the scaled frequencies scaled by $ 2 \pi \tau$. For every value of $\tau>0$, we define the scaled DFT of $\phi_{\tau}^{D}$ and its associated discrete power spectrum as
\begin{align}
\label{def:scaledDFT}
F^D_{\tau}(u_k) &\equiv \frac{\delta}{\sqrt{2 \pi \tau}} \sum_{n=0}^{N-1} \phi_{\tau}^{D}(y_n) \exp\left(-\frac{i u_k y_n}{\tau}\right) \\
\label{def:scaledDPS}
P^D_{\tau}(u_k) &\equiv \left|F^D_{\tau}(u_k)\right|^2.
\end{align}
The scale factor $\frac{1}{\sqrt{2 \pi \tau}}$ compensates for the scaling of the frequencies $w_k$ by $2 \pi \tau$ leading to the following lemma whose proof is given in Appendix~\ref{sec:proofofscaledPSlemma}.
\begin{lemma}
\label{lemma:scaledPS}
The scaled DFT in (\ref{def:scaledDFT}) and (\ref{def:scaledDPS}) satisfies $\frac{2 \pi \tau}{N \delta} \sum_{k=0}^{N-1} P^D_{\tau}(u_k) = 1$.
\end{lemma}
\subsection{Scaled Fourier transform}
\label{sec:scaledFT}
By defining the following constants
\begin{align}
\rho & = \frac{\delta}{2} \nonumber \\
\label{def:rho1}
\rho_1 &= -\rho\\
\label{def:rho2}
\rho_2 &= L+\rho,
\end{align}
we construct a continuous function $\Hdelta(x) : \mathbb{R} \rightarrow [0,1]$ as follows:
\begin{equation}
\label{def:H}
\Hdelta(x) \equiv \left\{ \begin{array}{ll}
              1 & \mbox{if $x\in [0, L]$}; \\ & \\
              0 & \mbox{if $x \leq \rho_1$ or $x \geq \rho_2$}; \\ & \\
              \frac{x-\rho_1}{-\rho_1} & \mbox{if $\rho_1 \leq x \leq 0$}; \\& \\
              \frac{\rho_2-x}{\rho_2-L} & \mbox{if $L \leq x \leq \rho_2$}.
          \end{array} \right.
\end{equation}
Denote $\Idelta = (\rho_1, 0) \bigcup (L,\rho_2)$ of length $\delta$ where $\Hdelta$ is linearly interpolated. The linear interpolation guarantees that in $\Idelta $, $|\Hdelta^{\prime}(x)| = \frac{2}{\delta}$, a constant, and $\Hdelta^{\prime \prime}(x) = 0$ which will prove useful in our proof. Note that though $\Hdelta$ is continuous everywhere, $\Hdelta^{\prime}$ is discontinuous at the limit points of $\Idelta$. Using one sided limits the derivatives of $\Hdelta$ can be appropriately extended to the limit points of $\Idelta$.
We then use $\Hdelta$ to define the sinusoidal function  $\phi_{\tau}(x): \mathbb{R} \rightarrow \mathbb{C}$ for all $\tau > 0$ as
\begin{equation}
\label{def:phi}
\phi_{\tau}(x) \equiv \Hdelta(x)  \frac{1}{\sqrt{L}} \exp\left(\frac{i S(x)}{\tau}\right).
\end{equation}
where we extend $S$ beyond the precincts of $[0,L]$ such that $s(x) = s(0), \forall x \in [\rho_1,0]$ and $s(x) = s(L), \forall x \in [L,\rho_2]$. As $\C$ includes $s(0)$ and $s(L)$, the aforementioned extension would ensure that $\forall u \in \mathbb{R}-\C$, $\Au \bigcap \Idelta = \emptyset$ as the interval $\Idelta$ is artificially introduced and should not interfere with the computation of the density. As $\Hdelta(x)$ is identically zero outside $[\rho_1,\rho_2]$ any extension of $S$ outside $[\rho_1,\rho_2]$ will not impact $\phi_{\tau}(x)$. Imposing $\Hdelta(x)=1$ between $y_0 = \frac{\delta}{2}$ and $y_{N-1} = L-\frac{\delta}{2}$ where the $N$ samples of $S$ are confined assures that
\begin{equation}
\label{eq:relationtophiD}
\phi_{\tau}(y_n) = \phi_{\tau}^{D}(y_n).
\end{equation}

Our subsequent analysis requires that $\phi_{\tau}(x)$ \emph{vanishes} at the end points $\rho_1$ and $\rho_2$ and also satisfies (\ref{eq:relationtophiD}) at the sample locations $y_n$. $\phi_{\tau}^{D}(0) \not= 0$ precludes us from setting $\rho = 0$. The non-zero choice of $\rho$ entails the introduction of the function $\Hdelta$ with properties as described in (\ref{def:H}). Setting $\rho=\frac{\delta}{2}$ and imposing $\Hdelta(x) = 0, \forall x \notin (\rho_1,\rho_2)$ forces $\phi_{\tau}\left((n+\frac{1}{2})\delta\right) = 0$ when $n \leq -1$ or $n \geq N$. This ensures that the Discrete Time Fourier Transform (DTFT) defined as
\begin{equation}
\label{def:DTFT}
F^{DTFT}(w) \equiv \delta \sum_{n=-\infty}^{\infty} \phi_{\tau}^{D}(y_n)  \exp(-i 2 \pi w y_n), \hspace{10pt} y_n = \left(n+\frac{1}{2}\right) \delta
\end{equation}
involving infinite summation \emph{coincides} with the DFT expression---stated in (\ref{def:DFT})---that comprises of only $N$ finitely many summations at the frequencies $w_k$, i.e., 
\begin{equation}
\label{eq:equalityDTFTDFT}
F^{DTFT}(w_k) = F^{D}(w_k), \forall k.
\end{equation}
In Section~\ref{sec:DFTFTrelation} we will see that enabling this equality will help us relate the DFT with the Fourier transform through the Poisson summation formula \cite{Elias71}. The constants $\rho_1$ and $\rho_2$ defined in (\ref{def:rho1}) and (\ref{def:rho2}) respectively ascertains that $\Hdelta(x)=1, \forall x \in \bigcup\limits_{k=0}^{N-1}\Auk$ thereby obstructing
$\Hdelta(x)$ from exercising any influence on the density of $s$ at the frequencies $u_k$. The definition of $\Hdelta(x)$ is totally left to our discretion and can be flexed to incorporate any desirable properties.

Counterpart to the discrete versions given in (\ref{def:DFT}) and (\ref{def:scaledDFT}) we define the Fourier transform (FT) and the scaled FT of $\phi_{\tau}(x)$ as
\begin{align}
\label{def:FT}
F(w) &= \int\limits_{\rho_1}^{\rho_2} \phi_{\tau}(x) \exp\left(-i 2 \pi w x\right) \, dx, \\
\label{def:scaledFT}
F_{\tau}(u) &= \frac{1}{\sqrt{2 \pi \tau}} \int\limits_{\rho_1}^{\rho_2} \phi_{\tau}(x) \exp\left(-\frac{i u x}{\tau}\right) \, dx
\end{align}
where again by relating $u = 2 \pi \tau w$ we get $F(w) = \sqrt{2 \pi \tau} F_{\tau}(u)$ akin to (\ref{eq:DFTscaledDFTrelation}).

\subsection{Relating the scaled DFT and the scaled Fourier transfrom}
\label{sec:DFTFTrelation}
The Poisson summation formula relates the DTFT with the Fourier transform ($F (w)$) where DTFT is just the periodic summation of $F(w)$ shifted by $\frac{1}{\delta}$ \cite{Elias71}. Using (\ref{eq:equalityDTFTDFT}), the Poisson summation formula can be leveraged to relate the DFT and the Fourier transform at these frequencies $w_k$, specifically
\begin{equation*}
F^D(w_k) = F^{DTFT}(w_k)= \sum_{l=-\infty}^{\infty} F\left(w_k - \frac{l}{\delta}\right).
\end{equation*}
Defining 
\begin{equation}
\label{def:gammal}
\gamma_l \equiv \frac{2 \pi \tau l}{\delta}
 \end{equation}
and using the scaled versions of the DFT and the Fourier transforms we get
\begin{equation}
\label{eq:relationDFTAndFT}
F^D_{\tau}(u_k) = F^{DTFT}_{\tau}(u_k)  = F_{\tau}(u_k) + \sum_{l=-\infty, l \not=0}^{\infty}F_{\tau}(u_k-\gamma_l)
\end{equation}
where $ F^{DTFT}_{\tau}(u)$ is the scaled DTFT given by
\begin{equation}
\label{def:scaledDTFT}
F^{DTFT}_{\tau}(u) \equiv \frac{F^{DTFT}(w)}{\sqrt{2 \pi \tau}}
\end{equation}
at $u = 2 \pi \tau w$ and is defined for all $u$. The infinite summation $\sum_{l=-\infty, l \not=0}^{\infty}F_{\tau}(u_k-\gamma_l)$ is known as the  \emph{aliasing error} \cite{Bracewell99}.

\section{Bound on $\tau$}
\label{sec:taubound}
As $s$ is also continuous on a compact interval $[\rho_1, \rho_2]$, the image of $s$ is also compact and hence bounded. Pick an arbitrarily small $\beta > 0$ and let 
\begin{equation}
\label{def:B}
B \equiv \sup\limits_{x \in [0,L]}|s(x)|+\beta
\end{equation}
such that $|s(x)|< B, \forall x \in [0,L]$. From (\ref{eq:freqwk}) note that $\max\limits_{k=0}^{N-1} |w_k| =  \frac{1}{2 \delta}$, hence $ \max\limits_{k=0}^{N-1} |u_k| =  \frac{\pi \tau}{\delta}$ and from the definition of $B$ in (\ref{def:B}) we have $\sup\limits_{x \in [0,L]}|s(x)| = B-\beta$ where $\beta>0$ is an arbitrarily small quantity. In the definition of $F^D_{\tau}(u_k)$ in (\ref{def:scaledDFT}) the frequencies $u_k$ are the histogram bins for the derivatives $s(x)$ where they are related by $u_{k}\leq s(y_n)<u_{k+1}$. Then, in order to \emph{capture} all the derivatives $\tau$ needs to be chosen such that
\begin{equation}
\label{eq:taubound}
\tau \geq \frac{B \delta}{\pi} = \left(\frac{B L}{\pi}\right) \frac{1}{N}.
\end{equation}
In the subsequent sections we will reason that the prudent choice of $\tau$ equals $\frac{B \delta}{\pi N}$. The \emph{linearity} of the relation between the free parameter $\tau$ and the sample interval $\delta$ at this value will prove crucial in obtaining the sought after $O(1/N)$ approximation of our density estimation technique. For now we let 
\begin{equation*}
\tau = \frac{C B \delta}{\pi}
\end{equation*}
for some constant $C \geq 1$.

\section{Bound on the aliasing error}
We start with the following lemma whose proof is given in Appendix~\ref{sec:proofofnostationarypointslemma}. Recall the definition of $T(x;u)$ in~(\ref{def:Txu}).
\begin{lemma}{[}No stationary points{]}
\label{lemma:nostationarypoints}
On the interval $[b_1,b_2] \subseteq [\rho_1,\rho_2]$ consider the integral
\begin{equation*}
W_{\tau}(u) = \int\limits_{b_1}^{b_2} \Hdelta(x) \exp\left(\frac{iT(x;u)}{\tau}\right) \, dx
\end{equation*}
under the condition that there exist a constant $\xi > 0$ such $|T^{\prime}(x;u)| \geq \xi , \forall x \in [b_1,b_2]$ implying the absence of any stationary points.
Then $W_{\tau}(u)  = O(\tau)$.
\end{lemma}

For $u \in \left[ \frac{-\pi \tau}{\delta}, \frac{\pi \tau}{\delta} \right]$ we now obtain a bound on the aliasing error as a function of $\tau$. When $\tau$ satisfies (\ref{eq:taubound}), $|\gamma_l| \geq 2 B |l|$ and
\begin{align}
|u-\gamma_l| \geq |\gamma_l| -|u| &\geq \frac{2 \pi \tau}{\delta} \left(|l|-\frac{1}{2}\right) \nonumber \\
\label{eq:diffgreaterbound}
&\geq 2 B \left(|l|-\frac{1}{2}\right) \geq B|l|.
\end{align}
Therefore,
\begin{align}
|T^{\prime}(x;u-\gamma_l)| &=|s(x)-(u-\gamma_l)| \nonumber \\ 
\label{eq:derivativebound}
						&\geq |u-\gamma_l| - |s(x)| \geq B(|l|-1) +\beta > 0
\end{align} 
indicating that the integral in computing $F_{\tau}(u_k-\gamma_l)$ is devoid of any stationary points. Applying Lemma~\ref{lemma:nostationarypoints} and recalling that
$\Hdelta(x)=0$ at the end points $\rho_1$ and $\rho_2$ by construction, the expression on the right side of (\ref{eq:Ifinal1}) vanishes. The remaining terms gives us
\begin{align}
\label{eq:aesummation1}
\sqrt{2 \pi L} F_{\tau}(u-\gamma_l) = & \frac{2\sqrt{\tau} CB}{\pi \left[s(0)-(u-\gamma_l)\right]^2} \left[\exp\left(\frac{i T(0;u-\gamma_l)}{\tau}\right)- \exp\left(\frac{iT(\rho_1;u-\gamma_l)}{\tau}\right)\right]\\
\label{eq:aesummation2}
&+\frac{2\sqrt{\tau} CB}{\pi \left[s(L)-(u-\gamma_l)\right]^2}\left[\exp\left(\frac{i T(L;u-\gamma_l)}{\tau}\right)-\exp\left(\frac{i T(\rho_2;u-\gamma_l)}{\tau}\right)\right]\\
\label{eq:aeerrorbound}
&+O\left(\frac{\tau \sqrt{\tau}}{[B(|l|-1) +\beta]^3}\right).
\end{align}
Realizing that
\begin{equation*}
\sum_{l=-\infty, l\not=0}^{\infty} \frac{1}{|l|^2} < \infty,
\end{equation*}
the \emph{infinite} summation of each of the term in (\ref{eq:aesummation1}) and (\ref{eq:aesummation2}) converges individually. The total aliasing error then satisfies
\begin{equation}
\label{bound:aliasingerror}
\sum_{l=-\infty, l \not=0}^{\infty} F_{\tau}(u-\gamma_l) = O(\sqrt{\tau}).
\end{equation}
In particular
\begin{equation}
\label{bound:aliasingerrorlthterm}
F_{\tau}(u-\gamma_l) = O\left(\frac{\sqrt{\tau}}{[B(|l|-1) +\beta]^2}\right).
\end{equation}

\section{Evaluation of Fourier transform via the method of stationary phase}
\label{sec:FTevaluation}
We now employ the stationary phase approximation technique \cite{OlverBook74,OlverArticle74} to obtain the asymptotic expression for the Fourier transform defined in (\ref{def:scaledFT}). We expand the scope of $\ut$ beyond the finite set of $N$ frequencies $\ut \in \{u_k\}_{k=0}^{N-1}$ where the scaled DFT values are defined to any $\tilde{u} \in  \mathbb{R} - \C$.

If no stationary point exits in $\Omega$ for the given $\tilde{u}$ ($|\Aut|=0$), then pursuant to Lemma~\ref{lemma:nostationarypoints} we have $F_{\tau}(\ut) = O(\sqrt{\tau})$. Otherwise, let the finite set $\Aut$ be represented by $\Aut=\{x_{1},x_{2},\ldots,x_{M(\ut)}\}$ with
$x_{m}<x_{m+1},\forall m$. We break $[\rho_1,\rho_2]$ into disjoint intervals
such that each interval has utmost one stationary point. To this end,
we choose numbers $\{c_0,c_{1},\ldots,c_{M(\ut)}\}$ such that
$\rho_1<c_{0}<x_{1}$, $x_{m}<c_{m}<x_{m+1}$ and $x_{M(\ut)}<c_{M(\ut)}<\rho_2$. We set $c_0 = 0$ and $c_{M(\ut)} =  L$ so that
the open interval $(c_0,c_{M(\ut)})$ encompasses all stationary points in $\Aut$. The choice of other constants will be discussed below. Recall that by definition $\Hdelta(x) = 1, \forall x \in [c_0, c_{M(\ut)}]$. The scaled Fourier transform $F_{\tau}(\ut)$ can be broken into:
\begin{equation}
F_{\tau}(\ut)\sqrt{2\pi\tau L}=G_1(\ut)+G_2(\ut)+\sum_{m=1}^{M(\ut)}K_m(\ut)+\tilde{K}_m(\ut)
\label{eq:Ftaubroken}
\end{equation}
where
\begin{align}
\label{def:G1}
G_{1,\tau}(\ut) & \equiv  \int_{\rho_1}^{c_0}\Hdelta(x)\exp\left(\frac{iT(x;\ut)}{\tau}\right)\, dx, \\
\label{def:G2}
G_{2,\tau}(\ut) &\equiv   \int_{c_{M(\ut)}}^{\rho_2}\Hdelta(x)\exp\left(\frac{iT(x;\ut)}{\tau}\right)\, dx, \\
\label{def:Km}
K_{m,\tau}(\ut) & \equiv   \int_{x_m}^{c_m}\exp\left(\frac{iT(x;\ut)}{\tau}\right) \, dx,\,\mathrm{and}\, \\
\label{def:tildaKm}
\tilde{K}_{m,\tau}(\ut) & \equiv   \int_{c_{m-1}}^{x_m}\exp\left(\frac{iT(x;\ut)}{\tau}\right)\, dx.
\end{align}
Evaluating  $K_{m,\tau}(\ut)$ and $\tilde{K}_{m,\tau}(\ut)$ using the method of stationary phase (\cite{OlverArticle74}, Chapter~3, Article~13 in \cite{OlverBook74}), (\ref{eq:Ftaubroken}) can be expressed as
\begin{align*}
F_{\tau}(\ut)\sqrt{2\pi\tau L} & =  \sum_{m=1}^{M(\ut)}\exp\left(\frac{i}{\tau}\left[S(x_{m})-\ut x_{m}\right]\right)\sqrt{\frac{2\pi\tau}{\left|S^{\prime\prime}(x_{m})\right|}}\exp\left(\pm\frac{i\pi}{4}\right)\nonumber \\
 & +\epsilon_{1,\tau}(\ut)+\epsilon_{2,\tau}(\ut).
 \end{align*}
 Depending on whether $S^{\prime \prime}(x_m)> 0$ or $<0$, the factor $\frac{i \pi}{4}$ in the exponent is positive or negative respectively.
 The error $\epsilon_{1,\tau}(\ut) \equiv G_{1,\tau}(\ut) + G_{2,\tau}(\ut)$ stems from computing the integral $F_{\tau}(\ut)$ on $[\rho_1,c_0] \bigcup [c_{M(\ut)}, \rho_2]$ which doesn\rq{}t contain any stationary points. Pursuant to Lemma~\ref{lemma:nostationarypoints}, $\epsilon_{1,\tau}(\ut) = O(\tau)$. Using the facts that $\Hdelta(x)=0$ at $x \in \{\rho_1, \rho_2\}$ and $\Hdelta(x)=1$ at $x \in \{c_0,c_{M(\ut)}\}$ in (\ref{eq:Ifinal1}), we get
 \begin{align}
 \label{eq:e1errorexpressed}
 \epsilon_{1,\tau}(\ut) = &-i\tau \frac{\exp\left(\frac{i T(c_0;\ut)}{\tau}\right)}{s(c_0)-\ut} + i\tau \frac{\exp\left(\frac{i T\left(c_{M(\ut)};\ut\right)}{\tau}\right)}{s\left(c_{M(\ut)}\right)-\ut} \nonumber \\
 &+\frac{2\tau CB}{\pi \left[s(0)-\ut \right]^2}\left[\exp\left(\frac{i T(0;\ut)}{\tau}\right) - \exp\left(\frac{i T(\rho_1;\ut)}{\tau}\right)\right] \nonumber \\
 &+\frac{2\tau CB}{\pi \left[s(L)-\ut \right]^2}\left[\exp\left(\frac{i T(L;\ut)}{\tau}\right) - \exp\left(\frac{i T(\rho_2;\ut)}{\tau}\right)\right]. \nonumber \\
 \end{align}
As the integral in $\epsilon_{1,\tau}(\ut)$ \emph{excludes} the interval $[0,L]$, the bound appearing in (\ref{eq:Ifinal4}) has been deliberately omitted.
 $\epsilon_{2,\tau}(\ut)$ represents the error from the stationary phase approximation and is derived to be (\ref{eq:e2errorexpressed}) in Appendix~\ref{sec:e2bound}. Using these error bounds we get
 \begin{align}
 \label{eq:Ftau_approx}
F_{\tau}(\ut) = \frac{1}{\sqrt{L}}\sum_{m=1}^{M(\ut)}\frac{\exp\left(\frac{i}{\tau}\left[S(x_m)-\ut x_m\right]\right)}{\sqrt{\left|S^{\prime\prime}(x_{m})\right|}}\exp\left(\pm\frac{i\pi}{4}\right)+\epsilon_{3,\tau}(\ut)
 \end{align}
 where 
 \begin{align}
 \label{eq:e3}
 \epsilon_{3,\tau}(\ut)=&\frac{\epsilon_{1,\tau}(\ut)+\epsilon_{2,\tau}(\ut)}{\sqrt{2\pi\tau L}} \\
 \label{eq:e3mainterm1}
 &=\frac{2\sqrt{\tau}CB}{\pi \sqrt{2 \pi L} \left[s(0)-\ut\right]^2}\left[\exp\left(\frac{i T(0;\ut)}{\tau}\right) - \exp\left(\frac{i T(\rho_1;\ut)}{\tau}\right)\right]\\
  \label{eq:e3mainterm2}
 &+\frac{2\sqrt{\tau}CB}{\pi \sqrt{2 \pi L} \left[s(L)-\ut\right]^2}\left[\exp\left(\frac{i T(L;\ut)}{\tau}\right) - \exp\left(\frac{i T(\rho_2;\ut)}{\tau}\right)\right] \\
\label{eq:e3error}
&+ O(\tau) \\
\label{bound:e3}
&=O(\sqrt{\tau}).
 \end{align}
To understand the bound of $O(\tau)$ for $ \epsilon_{3,\tau}(\ut)$ in (\ref{eq:e3error}) note that when we combine (\ref{eq:e1errorexpressed}) and (\ref{eq:e2errorexpressed}) to add the error terms $\epsilon_{1,\tau}(\ut)$ and $\epsilon_{2,\tau}(\ut)$, all the phase terms containing the constants $c_0$ and $c_{M(\ut)}$ \emph{cancel} each other. The remainder error term $O(\tau \sqrt{\tau})$ when divided by $\sqrt{\tau}$ appearing in the denominator of (\ref{eq:e3}) results in a bound of $O(\tau)$.
 
The scaled power spectrum $P_{\tau}(\ut) \equiv |F_{\tau}(\ut)|^2$ equals
\begin{align}
P_{\tau}(\ut) = & \frac{1}{L}\sum_{m=1}^{M(\ut)}\frac{1}{\left|S^{\prime\prime}(x_m)\right|} +\frac{1}{L}\sum_{m=1}^{M(u)}\sum_{t=1;t\not=m}^{M(\ut)}\chi_{m,t,\tau}(x_m,x_t,\ut) \nonumber \\
 \label{eq:Ph}
 & +\left|\epsilon_{3,\tau}(\ut)\right|^2+\epsilon_{4,\tau}(\ut) + \overline{\epsilon_{4,\tau}(\ut)}
 \end{align}
  where 
\begin{align}
\label{eq:crossxmxt}
\chi_{m,t,\tau}(x_m,x_t,\ut) &= \frac{\cos\left(\frac{1}{\tau}\left[S(x_m)-S(x_t)-\ut(x_m-x_t)\right]+\theta_{m,t}(x_m,x_t)\right)}{\sqrt{\left|S^{\prime\prime}(x_m)\right|}\sqrt{\left|S^{\prime\prime}(x_t)\right|}}, \\
\label{eq:e4}
\epsilon_{4,\tau}(\ut) &= \overline{\epsilon_{3,\tau}(\ut)} \frac{1}{\sqrt{L}}\sum_{m=1}^{M(\ut)}\frac{\exp\left(\frac{i}{\tau}\left[S(x_m)-\ut x_m\right]\right)}{\sqrt{\left|S^{\prime\prime}(x_{m})\right|}}\exp\left(\pm\frac{i\pi}{4}\right).
\end{align}
The cross terms $\chi_{m,t,\tau}(x_m,x_t,\ut)$ germinates from having multiple spatial locations ($x_m, x_t$) index into the same frequency bin $\ut$. 
Additionally, $\theta_{m,t}(x_m,x_t)=0,\,\frac{\pi}{2}$ or $-\frac{\pi}{2}$ and $\theta_{m,t}(x_t,x_m)=-\theta_{m,t}(x_m,x_t)$.

\section{Approximation error of our wave function method}
To keep up with our analysis for any $\ut \in \mathbb{R}-\C$ rather than confining to the set of $N$ scaled frequencies $\{u_k\}_{k=0}^{N-1}$, we use the scaled DTFT instead of the scaled DFT. Let $P^{DTFT}_{\tau}(\ut)$ represent the magnitude square of the scaled DTFT.  Substituting $\tau = \frac{C B \delta}{\pi}$, observe that the scaled frequencies lie between $\left[-CB, CB\right]$ for all $N$ where $C \geq 1$.  Additionally, as $|s(x)| < B, \forall x \in \Omega$ by definition, the true density $P(\tilde{u})=0, \forall \tilde{u} \notin (-B,B)$. So we restrict ourselves to the interesting region where $\ut \in \left[-CB, CB\right] - \C$. Recollect that we explicitly avoid the set $\C$ where the density $P(\tilde{u})$ is not defined. The formal mathematical statement of our result can be stated as follows:
\begin{theorem}
\label{thm:derivativedensity}
For any $\tilde{u} \in  [-B,B] - \C$, there exists a closed interval $\mathcal{N}_{\alpha}(\tilde{u}) = [\tilde{u}-\tilde{\alpha}, \tilde{u}+\tilde{\alpha}]$ with $\tilde{\alpha}$ chosen independent of $N$---as given by Lemma~\ref{lemma:neighborhoodLemma}---such that when $\tau =  \left(\frac{B L}{\pi}\right) \frac{1}{N}$, the cumulative of the difference $P^{DTFT}_{\tau}(u) - P(u)$ over $\Nalphat(\ut)$ is of $O(1/N)$ as $N \rightarrow \infty$, i.e,
\begin{equation}
\label{eq:integralconvergence}
\int\limits_{\Nalphat(\ut)}\left[P^{DTFT}_{\tau}(u)- P(u)\right]\, du =  O\left(\frac{1}{N}\right).
\end{equation}
\end{theorem}
\begin{proof}
\textbf{case (i)} No stationary points: Plugging the $O(\sqrt{\tau})$ bound of the aliasing error and the Fourier transform in (\ref{eq:relationDFTAndFT}) and taking the magnitude square we get $P^{DTFT}_{\tau}(\ut) = O(\tau)$. As $\ut \notin s([\rho_1, \rho_2])$ and the image $s([\rho_1, \rho_2])$ is compact, there exist a neighborhood $\Nalphat(\ut)$ around $\ut$ such that $\forall u \in \Nalphat(\ut)$, no stationary points exists. The selection of $\tau$ as a function of $N$ is discussed below where we reason that the judicious choice of $\tau$ is $\left(\frac{B L}{\pi}\right) \frac{1}{N}$. By integrating $P^{DTFT}_{\tau}(\ut)$ over $\Nalphat(\ut)$ the result follows.\\\\
\textbf{case (ii)} Existence of stationary points: Considering the magnitude square of (\ref{eq:relationDFTAndFT}) and plugging in (\ref{eq:Ph}) we get
\begin{equation}
\label{eq:DPSrelation}
P^{DTFT}_{\tau}(\ut) = P(\ut) + \epsilon_{5,\tau}(\ut)
\end{equation}
where
\begin{align}
\label{eq:errorct}
\epsilon_{5,\tau}(\ut) = &\frac{1}{L}\sum_{m=1}^{M(\ut)}\sum_{t=1;t\not=m}^{M(\ut)}\chi_{m,t,\tau}(x_m,x_t,\ut) \\
&+\left|\epsilon_{3,\tau}(\ut)\right|^2+\epsilon_{4,\tau}(\ut) + \overline{\epsilon_{4,\tau}(\ut)}\\
\label{eq:aesq}
&+\left| \sum_{l=-\infty, l \not=0}^{\infty}F_{\tau}(\ut-\gamma_l)\right|^2\\
\label{eq:aecross}
& + \overline{F_{\tau}(\ut)} \left(\sum_{l=-\infty, l \not=0}^{\infty}F_{\tau}(\ut-\gamma_l)\right) + F_{\tau}(\ut) \overline{\left(\sum_{l=-\infty, l \not=0}^{\infty}F_{\tau}(\ut-\gamma_l)\right)}.
\end{align}
Based on the form of the cross terms in (\ref{eq:crossxmxt}) it is straightforward to check that $\lim\limits_{\tau \rightarrow 0} \chi_{m,t,\tau}(x_m,x_t,\ut)$ doesn\rq{}t exist. Hence in order to recover the density we must integrate the power spectrum over an arbitrarily small neighborhood $\Nalphat(\ut)$ around $\ut$ to nullify these cross terms. Lemma~\ref{lemma:neighborhoodLemma} endows us with one such neighborhood. Recall that from Lemma~\ref{lemma:neighborhoodLemma}, $s^{(-1)}\left( \Nalphat(\ut)\right) = \bigcup\limits_{m=1}^{M(\ut)}\Nalphat(x_m)$ where $\Nalphat(x_m)$ is the image of $s^{(-1)}\left( \Nalphat(\ut)\right)$ confined around $x_m$. We set $\Nalphat(\ut) = [\ut - \tilde{\alpha}, \ut+\tilde{\alpha}]$ where we select a small enough $\tilde{\alpha}$ (independent of $\tau$) in accordance with Lemma~\ref{lemma:neighborhoodLemma} and also choose the remaining constants $\{c_1,c_2,\cdots,c_{M(\ut)-1}\}$ such that $\Nalphat(x_m) \subset (c_{m-1}, c_{m}), \forall m$. This would enable the definitions given in (\ref{def:G1})-(\ref{def:tildaKm}) concerning these constants to be extended $\forall u \in \Nalphat(\ut)$. Additionally, as $|s(x)| \leq B-\beta, \forall x \in [0,L]$ we further have
$\Nalphat(\ut) \subset (-B,B)$. The following two lemmas capture the $O(1/N)$ approximation of our density estimation method. Their proofs are provided in Appendices~\ref{sec:proofoflemma:productnostationarypoints} and \ref{sec:proofoflemma:boundsonintegrals} respectively.

\begin{lemma}
\label{lemma:productnostationarypoints}
Let the constant $\kappa \in \{\rho_1,0,L,\rho_2\}$ so that $s(\kappa) \not= u,\forall u \in \Nalphat(\ut)$. Let $|s(\kappa)-(u-\gamma)|\geq \xi$ for some constant $\xi > 0$ where $\gamma = \frac{2 \pi \tau l}{\delta} = 2 B C l , l \in \mathbb{Z}$ (including $l=0$). Define
\begin{equation}
\label{def:zeta}
\zeta_{\tau}(u)  \equiv \frac{\exp\left(\frac{iT(\kappa;u-\gamma)}{\tau}\right)}{\left[s(\kappa)-(u-\gamma)\right]^2}
\end{equation}
for $u \in \Nalphat(\ut)$.
Then
\begin{equation*}
Z_{\tau} = \int\limits_{\Nalphat(\ut)} \overline{\zeta_{\tau}(u)} \frac{\exp\left(\frac{i}{\tau}\left[S(x_m(u))-u x_m(u)\right]\right)}{\sqrt{\left|S^{\prime\prime}(x_m(u))\right|}}\, du = O\left(\frac{\tau}{\xi^2}\right), \forall m.
\end{equation*}
\end{lemma}

\begin{lemma}{[}Bound on Integrated Error Lemma{]}
\label{lemma:boundsonintegrals}
The bound on the each of the error terms in $\epsilon_{5,\tau}(u)$ when integrated over an interval $\Nalphat(\ut)$ chosen independent of $\tau$ are as summarized in Table~\ref{table:integral bounds} based on which we could conclude that
\begin{equation*}
\int\limits_{\Nalphat(\ut)} \epsilon_{5,\tau}(u) \, du = O(\tau).
\end{equation*}
\end{lemma}
\begin{center}
\begin{table}[ht!]
\caption{Bound on the integrals \label{table:integral bounds}}
\noindent \centering{}%
\begin{tabular}{|l|l|}
\hline 
Integrated over $\Nalphat(\ut)$  & Bound \tabularnewline
\hline\hline
$\chi_{m,t,\tau}\left(x_m(u),x_t(u),u\right)$ & $O(\tau)$ \tabularnewline
\hline 
$\left|\epsilon_{3,\tau}(u)\right|^2$  & $O(\tau)$ \tabularnewline
\hline 
$\epsilon_{4,\tau}(u)$ & $O\left(\tau\right)$ \tabularnewline
\hline 
$\left| \sum_{l=-\infty, l \not=0}^{\infty}F_{\tau}(u-\gamma_l)\right|^2$ & $O(\tau)$ \tabularnewline
\hline 
$\overline{F_{\tau}(u)} \left(\sum_{l=-\infty, l \not=0}^{\infty}F_{\tau}(u-\gamma_l)\right)$ & $O(\tau)$\tabularnewline
\hline 
\end{tabular}
\end{table}
\par\end{center}

\noindent \textbf{Choice of $\tau$ as a function of $N$:}
In Section~\ref{sec:taubound} we demonstrated that if there are only $N$ finitely many samples of $S$ picked at intervals of $\delta$, $\tau$ cannot be set arbitrarily close to zero and should respect the inequality~(\ref{sec:taubound}). Lemma~\ref{lemma:boundsonintegrals} establishes that the integral of the error $\epsilon_{5,\tau}(u)$ between the true and the estimated density over a small interval $\Nalphat(\ut)$ is bounded by $\tau$ and hence we can expect the error profile to portray a decreasing trend as we tune down $\tau$. Apropos to the aforementioned statements it is logical to conclude that the judicious choice of $\tau$ for a given $N$ equals
\begin{equation*}
\tau = \frac{B \delta}{\pi} = \left(\frac{B L}{\pi}\right) \frac{1}{N}.
\end{equation*}
The inverse relation between $\tau$ and $N$ proves Theorem~\ref{thm:derivativedensity}.
\end{proof}

We also obtain the following corollary as a direct consequence of Theorem~\ref{thm:derivativedensity}.
\begin{corollary}
\label{cor:derivativedensity}
For all $\tilde{u} \in  [-B,B] - \C$ consider the closed interval $\mathcal{N}_{\alpha}(\tilde{u}) = [\tilde{u}-\tilde{\alpha}, \tilde{u}+\tilde{\alpha}]$ of length $2 \tilde{\alpha}$ satisfying Lemma~\ref{lemma:neighborhoodLemma}. Then
\begin{equation}
\label{eq:densityconvergence}
\lim\limits_{\tilde{\alpha}\rightarrow 0} \frac{1}{2 \tilde{\alpha}} \lim\limits_{N \rightarrow \infty} \int\limits_{\Nalphat(\ut)}P^{DTFT}_{\left(\frac{B L}{\pi N}\right)}(u) = P(\ut).
\end{equation}
\end{corollary}

\section{Experimental justification}
\label{sec:expverification}
Below we experimentally verify the upper bound on the approximation error of our gradient density estimator for some randomly chosen different types of functions $S(x)$ which are polynomials, exponentials, sinusoids, logarithmic and combinations of these. The sampling locations $\{y_n\}_{n=1}^N$ where chosen between the interval $[0,2]$ for different values of interval width $\delta = \frac{2}{N}$. The bound $B$ on the derivative was approximated by setting it to the maximum absolute of the derivative computed via finite differences, i.e,
\begin{equation}
\label{eq:Bapprox}
 B \sim \max\limits_{n=0}^{N-2} \left|\frac{S\left(y_{n+1}\right)- S\left(y_n\right)}{\delta}\right|.
\end{equation}
The true density $P(u)$ was either computed in closed form whenever possible or approximated via standard histograms. We set $\tau$ to its corresponding lower bound $\frac{BL}{\pi N}$. 

To showcase the efficacy of our wave function method, in the left panel of Figure~\ref{fig:diffFuncs} we plot the true density and in the right panel we show the estimated density via our stationary phase method computed at $K=2048$ frequency values using $N=65536$ samples. It is visually clear that the density determined from our wave function method is almost identical to the true density function.  Also notice that the frequency locations where the true density is zero, our estimated density function is also zero. Further, we investigated two variations of our convergence results as follows.

 \begin{figure}
 \centering
 \includegraphics[width=0.49\textwidth]{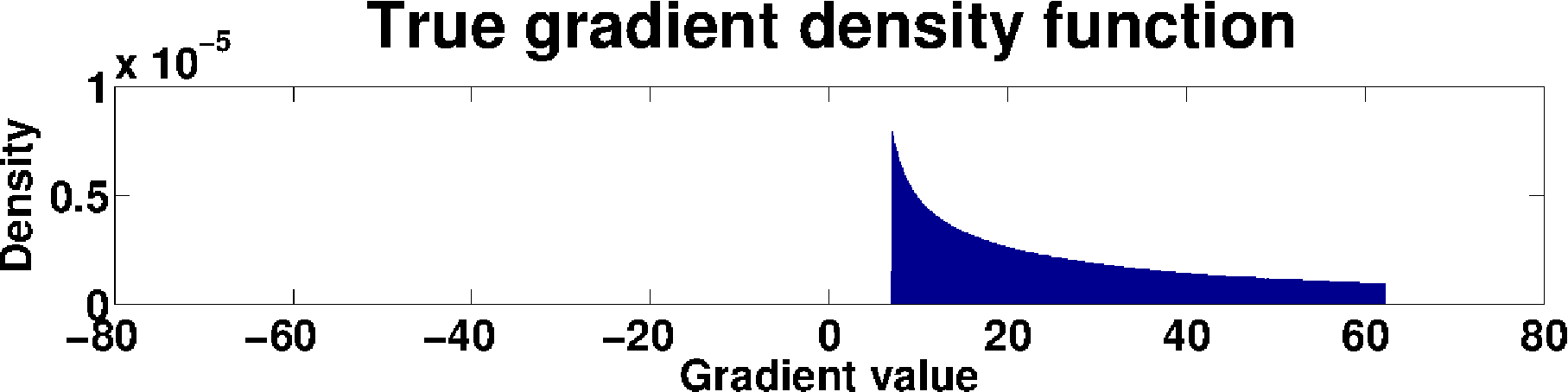} 
  \includegraphics[width=0.49\textwidth]{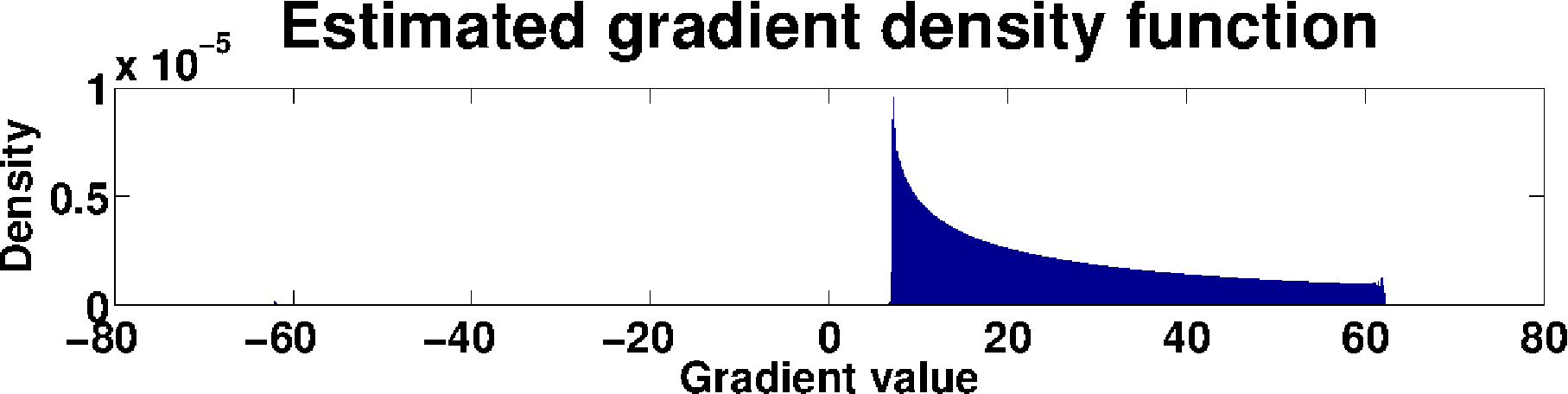}
  \includegraphics[width=0.49\textwidth]{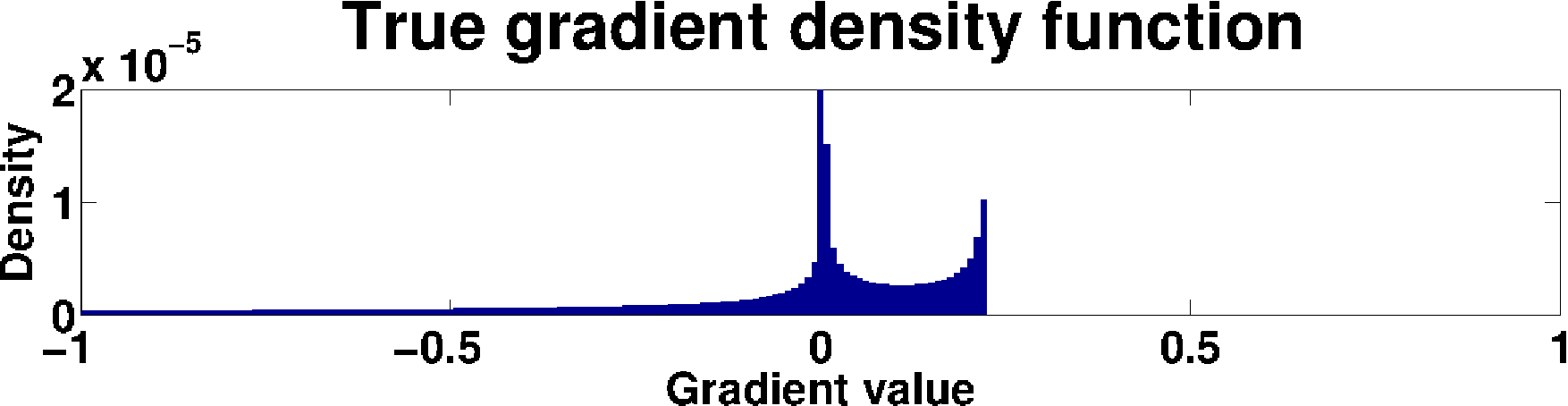} 
  \includegraphics[width=0.49\textwidth]{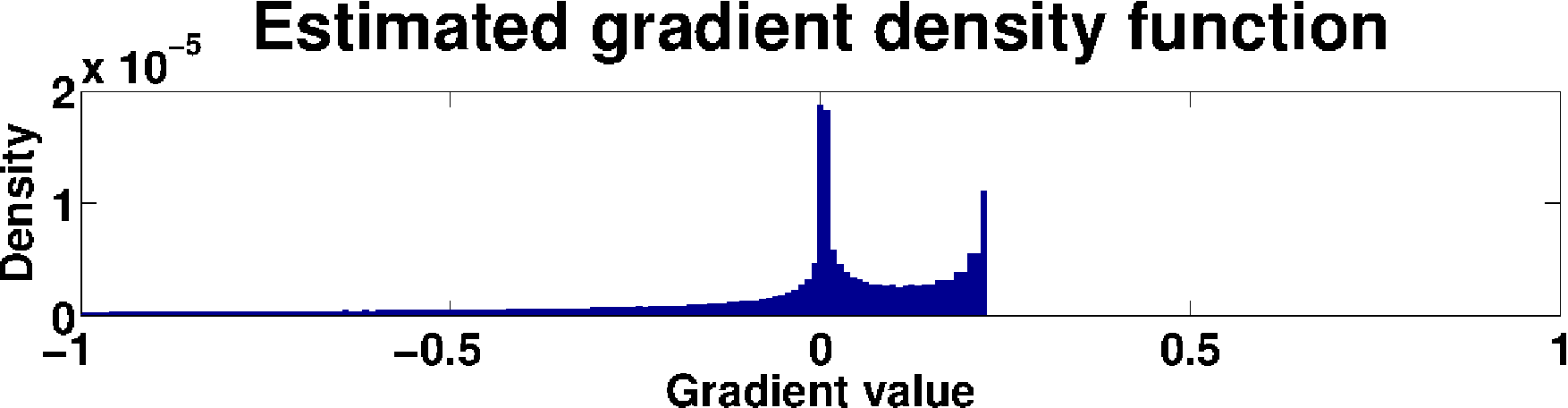}
   \includegraphics[width=0.49\textwidth]{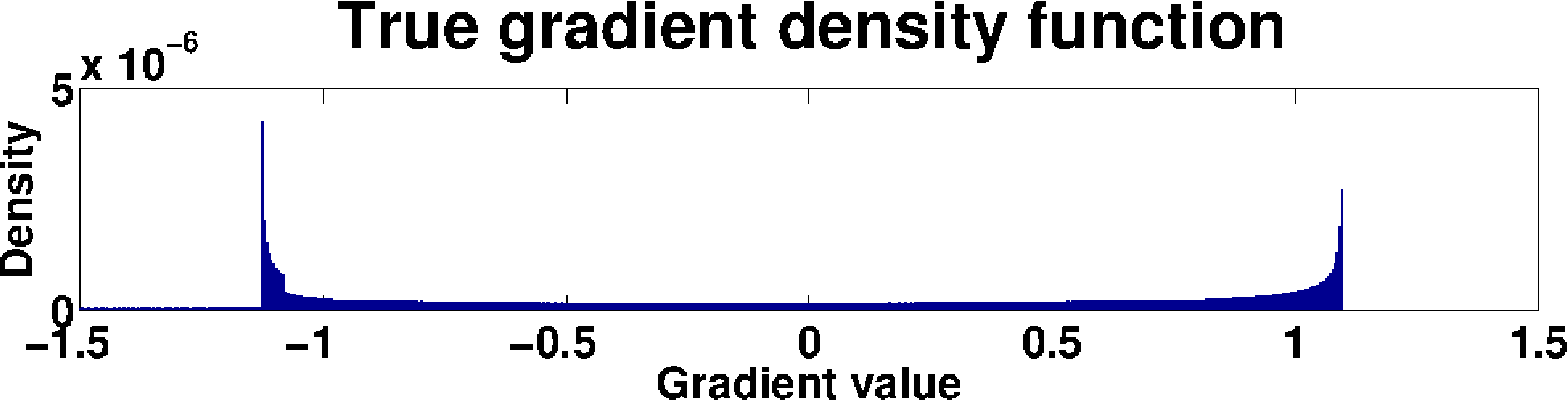} 
  \includegraphics[width=0.49\textwidth]{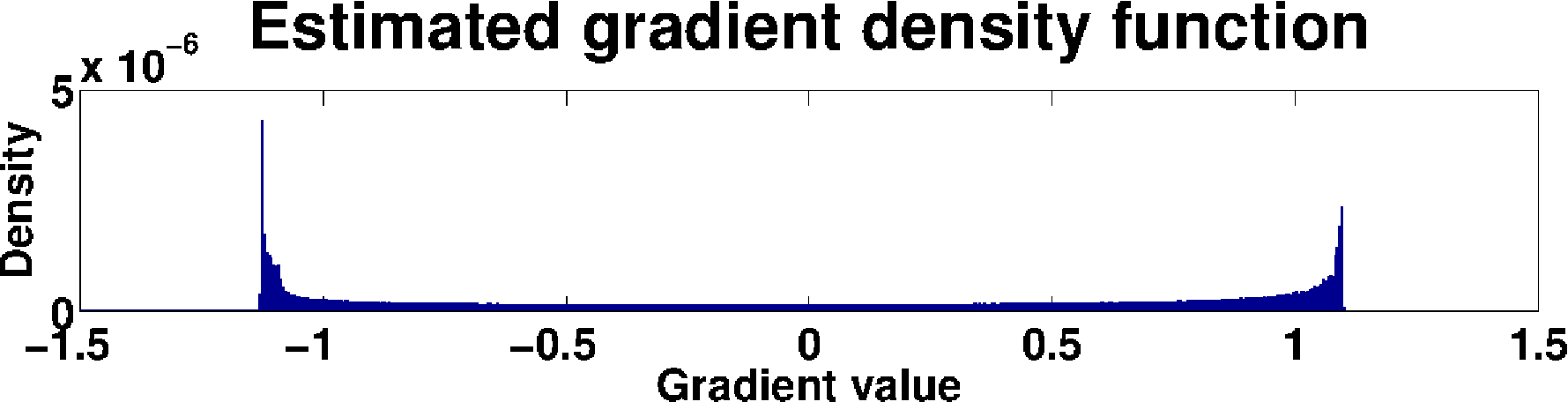}
  \includegraphics[width=0.49\textwidth]{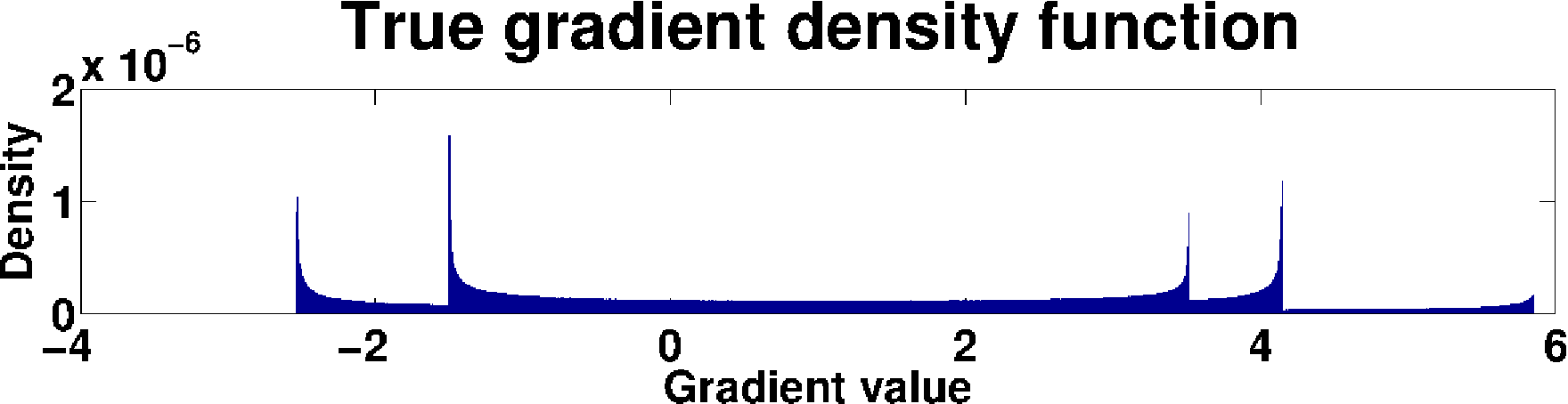} 
  \includegraphics[width=0.49\textwidth]{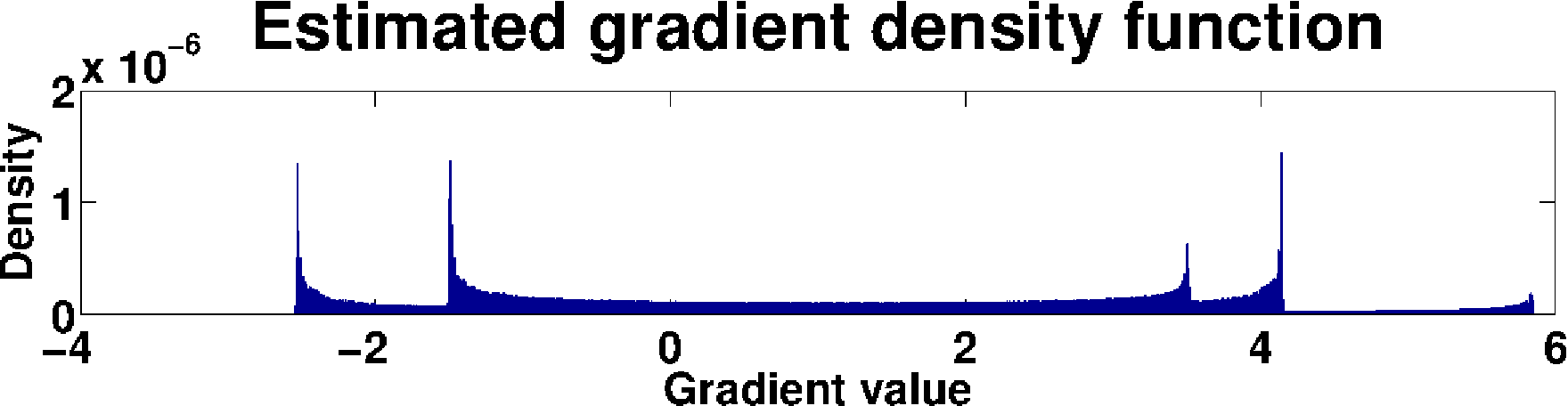}  
  \caption{Comparison between (i) Left: True gradient density function, and (ii) Right: Estimated gradient density function.}
   \label{fig:diffFuncs}
  \end{figure}
\subsection{Case study 1} \emph{Progressively increase the the number of samples $N$ and for each $N$ set $\tau$ to its corresponding lower bound $\frac{BL}{\pi N}$.}

Our Theorem~\ref{thm:derivativedensity} require both the power spectrum and the true density to be integrated over a small neighborhood $\Nalpha$ chosen independent of $N$ to observe convergence. To this end we preselect a set of $K = 255$ fixed frequencies $\{\ut_1, \ut_2,\cdots,\ut_{K}\}$ and consider appropriate non-overlapping neighborhoods $\{\Nalphat(\ut_k)\}_{k=1}^K$ around them in accordance with Lemma~\ref{lemma:neighborhoodLemma}. Note that as we scale up $N$, the \emph{number} of frequency locations $N_k$ within each neighborhood $\Nalphat(\ut_k)$ where the discrete power spectrum is defined also increases and hence we could progressively approximate the integrals in Theorem~\ref{thm:derivativedensity} by its summation involving $P^{DFT}_{\tau}(u)$. The replacement of the integral with sums is akin to the well known Riemann summation approximation though not exactly equivalent as the underlying function $P^{DTFT}_{\tau}(u)$ whose samples we get in the form of $P^{DFT}_{\tau}(u)$ keeps varying with $N$ (and also with $\tau$) as easily seen from the definition of the scaled DTFT in (\ref{def:scaledDFT}) that involves summation over $N$ terms. Notwithstanding this conceptual difference and continuing to term it as Riemann sums we define the error between their respective Riemann summation as
\begin{equation}
\label{eq:DeltatauN}
\Delta_{\tau,N} = \frac{1}{K} \sum_{k=1}^K \frac{2 \pi \tau}{N \delta}\left|\sum_{l=1}^{N_k} \left[P_{\tau}^D\left(u_{l,k}\right) - P\left(u_{l,k}\right)\right] \right|
\end{equation}
where $u_{l,k}$ is the $l^{th}$ frequency location within the interval $\Nalphat(\ut_k)$. The spacing between the frequencies equals $\frac{2 \pi \tau}{N \delta}$. In Figure~\ref{fig:densityconvergence} we visualize these Riemann summations where we observe that as we increase $N$, the interval measure of our density function (plotted right) steadily approaches the interval measure of the true density function (plotted left) corroborating our assertion that the power spectrum can increasingly, accurately serve as the gradient density estimator for large values of $N$. In Figure~\ref{fig:convergencerate} we plot $\Delta_{\tau,N}$ for different values of $\tau$ and find it to be \emph{linear} ascertaining our Theorem~\ref{thm:derivativedensity}.

\subsection{Case study 2} \emph{Fix $N = N_0$ and progressively decrease $\tau$ from some high value to its appropriate choice $\tau_0 = \delta_0$.}

The purpose of this case study is to verify that the lower bound on $\tau = \frac{B L}{\pi N}$ is indeed its optimum value. We fixed $N=N_0=65536$ and computed the average summation error $\Delta_{\tau,N_0}$ according to (\ref{eq:DeltatauN}) for varying values of $\tau$, averaged over the preselected $K=255$ fixed number of frequencies. The plot in Figure~\ref{fig:optimumtau} displays the behavior of $\Delta_{\tau,N_0}$ with $\tau$. Note that the number of samples $N_k$ within each neighborhood $\Nalphat(\ut_k)$ does not change as $N$ is held constant. However, the values of discrete power spectrum $P^{DFT}_{\tau}(u_{l,k})$ at the frequencies $u_{l,k}$ varies with changing $\tau$.
The following inferences can be deduced from the profile of the graph in Figure~\ref{fig:optimumtau} namely:
\begin{enumerate}
\item The error steadily decreases with $\tau$ as we approach its lower bound.
\item The rate of decline is almost \emph{linear} in $\tau$ substantiating the concluding remarks of Lemma~\ref{lemma:boundsonintegrals}.
\end{enumerate}
 \begin{figure}
 \centering
 \includegraphics[width=0.49\textwidth]{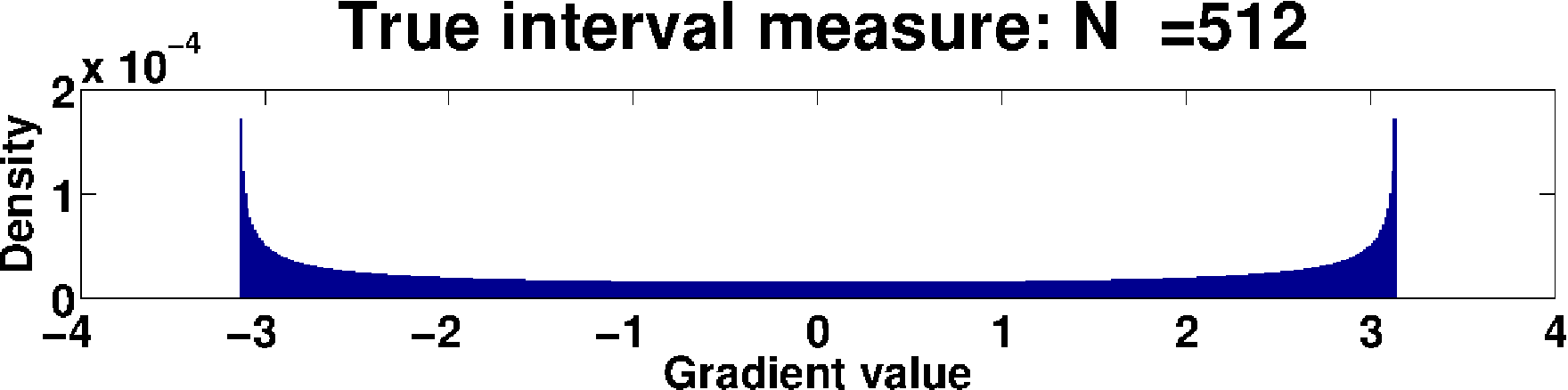} 
  \includegraphics[width=0.49\textwidth]{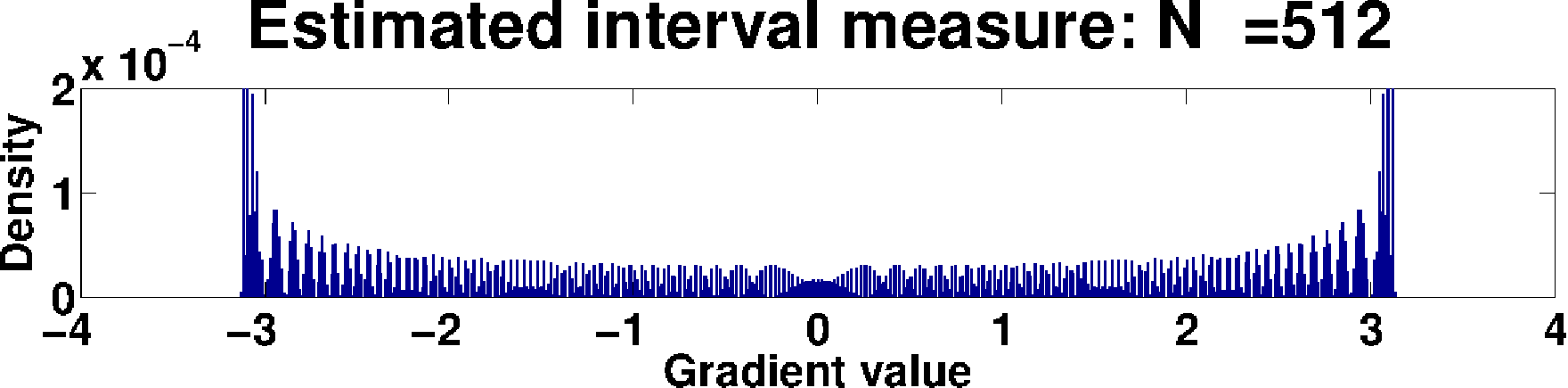} 
 \includegraphics[width=0.49\textwidth]{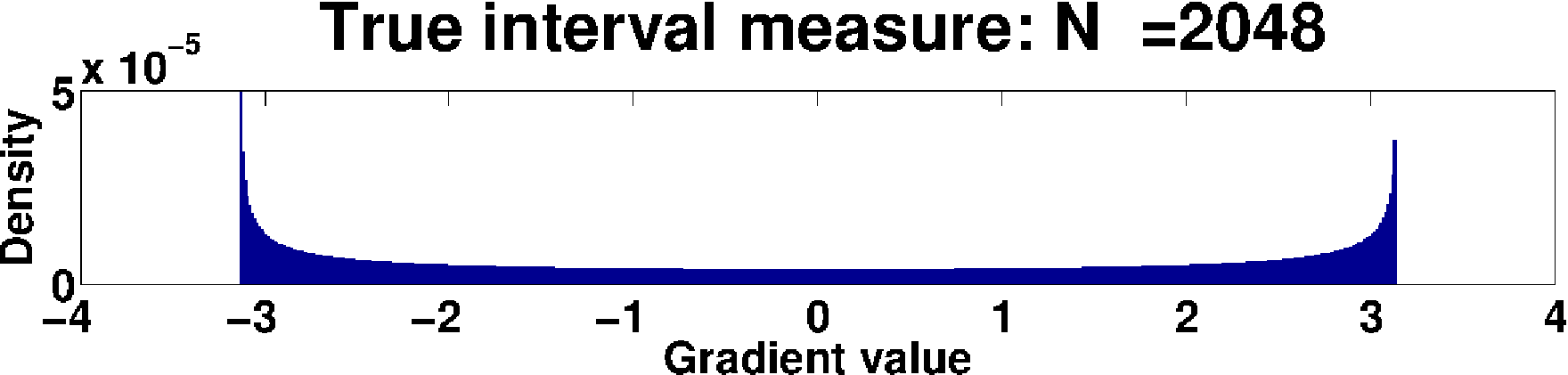}
 \includegraphics[width=0.49\textwidth]{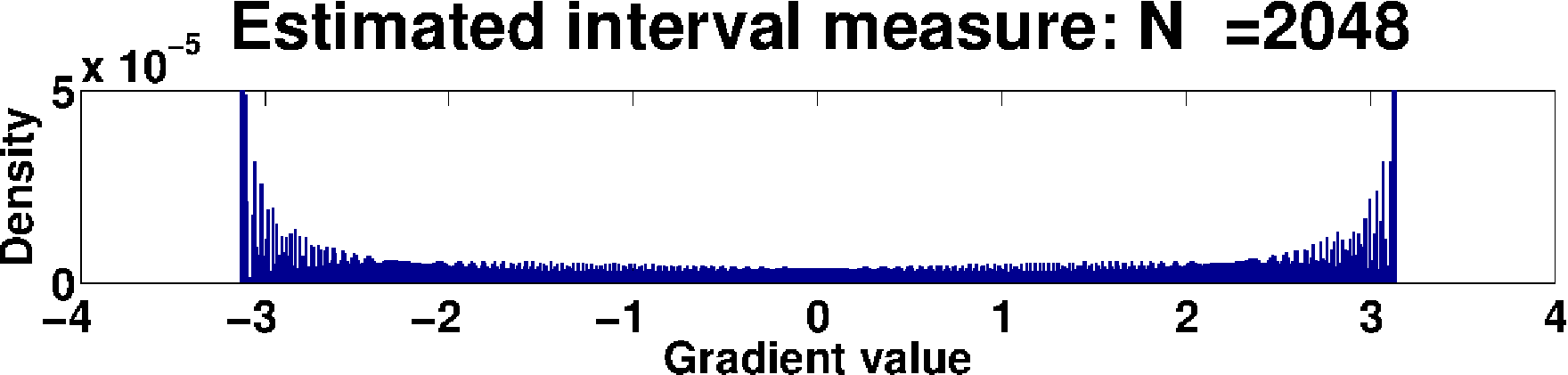} 
  \includegraphics[width=0.49\textwidth]{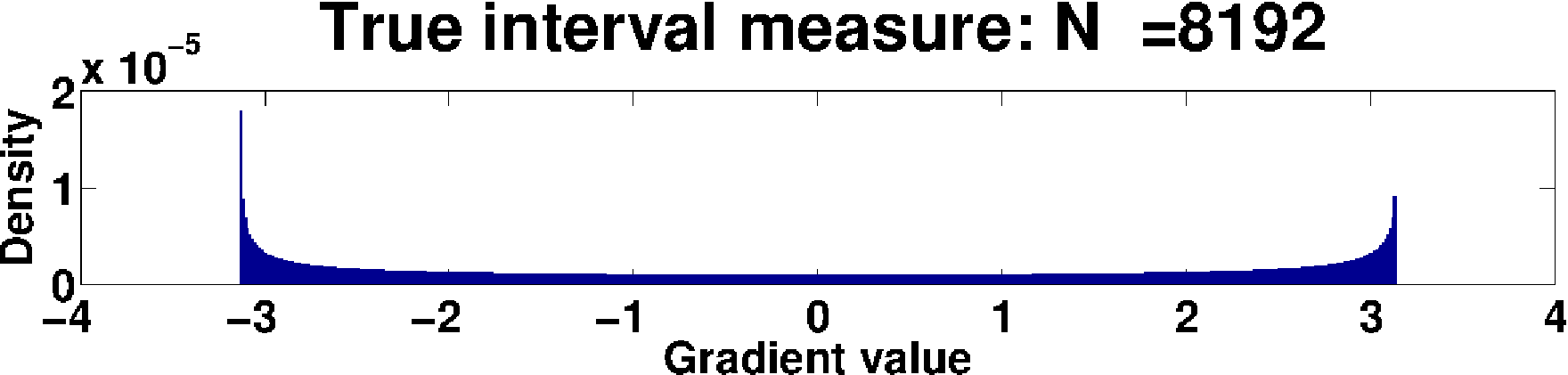}
  \includegraphics[width=0.49\textwidth]{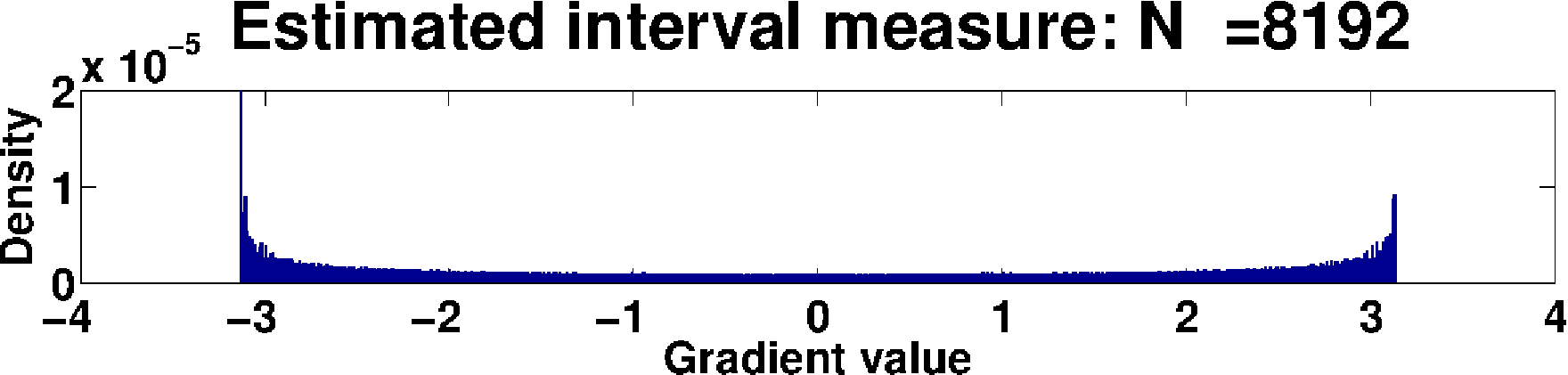} 
  \includegraphics[width=0.49\textwidth]{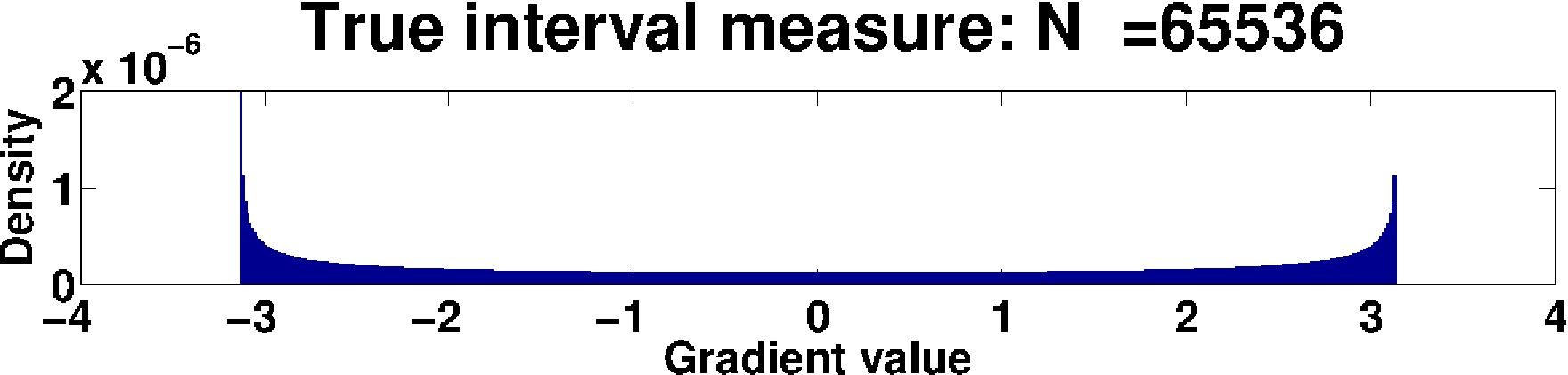}
  \includegraphics[width=0.49\textwidth]{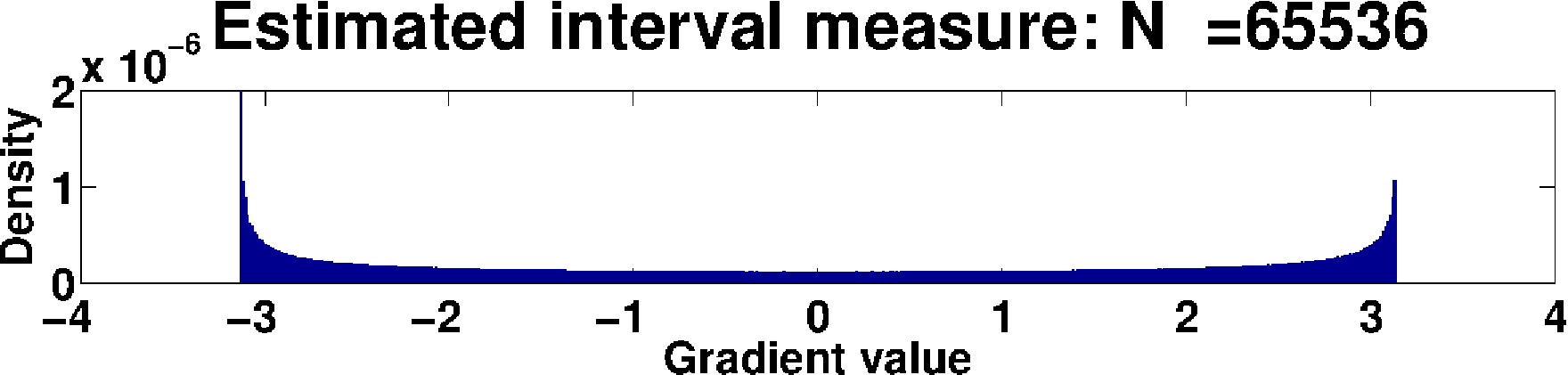} 
  \caption{Convergence with increasing $N$. (i) Left: True interval measure, and (ii) Right: Estimated interval measure.}
   \label{fig:densityconvergence}
  \end{figure}
  
  \begin{figure}
 \centering
 \begin{minipage}{0.49\textwidth}
 	\centering
	\includegraphics[width=1\textwidth]{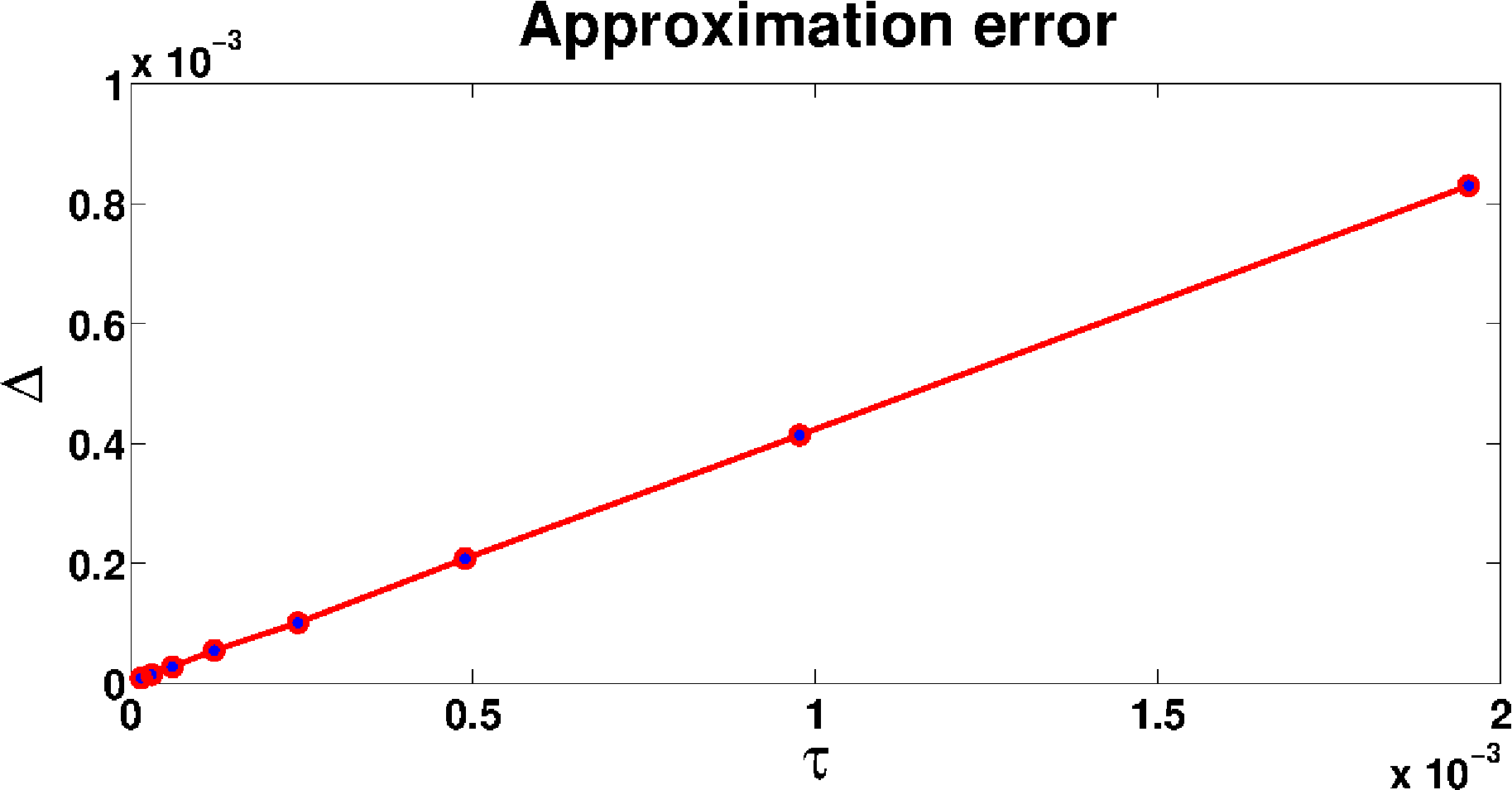}
	\caption{Approximation error with decreasing (increasing) $\tau$ ($N$).}
	\label{fig:convergencerate}
  \end{minipage}
  \hfill
  \begin{minipage}{0.49\textwidth}
 	\centering
	\includegraphics[width=1\textwidth]{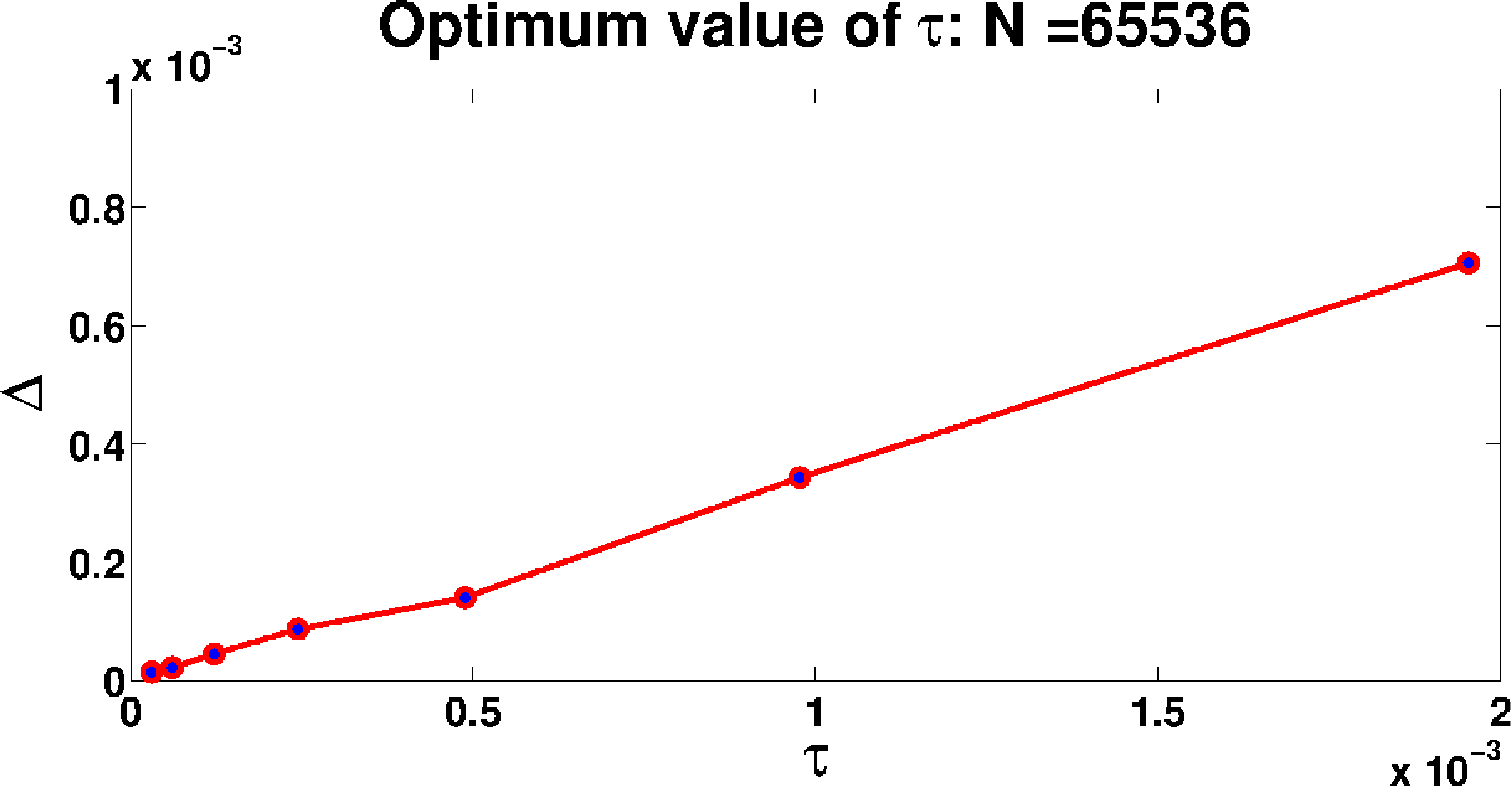}
	\caption{Approximation error with decreasing $\tau$ for a fixed $N=N_0$.}
	\label{fig:optimumtau}
  \end{minipage}
  \end{figure}
  
\section{Discussion}
The integrals
\begin{equation*}
I_{\tau}(\ut)= \int\limits_{\Nalphat(\ut)}P^{DTFT}_\tau (u) \, du, \hspace{10pt}I(\ut)=\int\limits_{\Nalphat(\ut)} P(u) \, du
\end{equation*}
 represent the interval measures of the density functions $P^{DTFT}_{\tau}$ and
$P$ respectively over an interval $\Nalphat(\ut)$ where the interval length can be made arbitrarily smaller but independent of $N$. Theorem~\ref{thm:derivativedensity} states that given the $N$ samples of $\phi(x)=\frac{1}{\sqrt{L}}\exp\left(\frac{iS(x)}{\tau}\right)$ and when $\tau$ is set to its lower bound $\frac{B L}{\pi N}$, both interval measures are almost equal with the difference between them decreasing at the fast rate of $O(1/N)$. Recall that the scaled discrete power spectrum $P^{D}_{\tau}(u_k)$ computed at the $N$ scaled frequencies $\{u_k\}_{k=0}^{N-1}$ spaced at increasingly smaller intervals of $\frac{2 \pi \tau}{N \delta} = \frac{2 B}{N}$ are the uniform samples of the discrete time power spectrum $P^{DTFT}_\tau (u)$. In Section~\ref{sec:expverification} we showed through simulations that the error between the Riemann sum approximation of interval measure $I_{\tau}(\ut)$ computed using $P^{D}_{\tau}(u_k)$ and those of $I(\ut)$ is bounded above by $O(1/N)$ and hence the discrete power spectrum can serve as the \emph{density estimator} for the derivative of $S$ at large values of $N$. Extension of this result to higher dimensions is a fruitful topic for future research. 

\appendix
\section{Proof of Lemmas}
\label{sec:proofoflemmas}
Below we provide the proofs for all the lemmas stated in this article.
\subsection{Proof of Finiteness Lemma}
\label{sec:proofoffinitenesslemma}
\begin{proof}
 We prove the result by contradiction. Observe that $\Au$
is a subset of the compact set $\Omega$. If $\Au$ is
not finite, then by Theorem~2.37 in \cite{Rudin76}, $\Au$
has a limit point $x_0 \in\Omega$. Consider a sequence $\{x_n\}_{n=1}^{\infty}$,
with each $x_{n}\in\Au$, converging to $x_{0}$. Since
$s(x_{n})=u,\forall n$, from the continuity of $s$
we get $s(x_{0})=u$ and hence $x_{0}\in\mathcal{A}_{u}$. Since $u\in\mathcal{C}$, $x_0 \notin \{0,L\}$. Additionally,
\begin{equation*}
\lim_{n\rightarrow\infty}\frac{s(x_0)-s(x_n)}{x_0-x_n}=0=s^{\prime}(x_0)=S^{\prime\prime}(x_{0})
\end{equation*}
 implying that $x_0 \in \mathcal{B}$ and $u \in \C$ resulting in a contradiction.
\end{proof}

\subsection{Proof of Neighborhood Lemma}
\label{sec:proofofneighborhoodlemma}
\begin{proof}
 Observe that $\mathcal{B}$ is closed---and being a subset of $\Omega$ is also compact---because if $x_{0}$ is a
limit point of $\mathcal{B}$, from the continuity of $S^{\prime\prime}$
we have $S^{\prime\prime}(x_{0})=0$ and hence $x_0 \in\mathcal{B}$.
Since $s$ is continuous, the set $\C$ is also compact and hence $\mathbb{R}-\mathcal{C}$ is open. Then
for $u\notin\mathcal{C}$, there exists an open neighborhood $\mathcal{N}_{r}(u)$
for some $r>0$ around $u$ such that $\mathcal{N}_{r}(u)\cap\mathcal{C}=\emptyset$.
By defining $\alpha=\frac{r}{2}$ ,we get the required closed neighborhood $\Nalpha$ containing $u$. 

Since $S^{\prime\prime}(x)$ is continuous and does not vanish $\forall x \in \Au$, all the other points of this lemma follow directly from the inverse function theorem. As $M(u)$ is finite by Lemma~\ref{lemma:finitenessLemma}, the neighborhood $\Nalpha$ can be chosen independently of $x \in \Au$ so that the points 1 and 3 are satisfied $\forall x \in \Au$. 
\end{proof}

\subsection {Proof of Density Lemma}
\label{sec:proofofdenistylemma}
\begin{proof}
 Since the random variable $X$ is assumed to have a uniform distribution
on $\Omega$ its density is given by $f_{X}(x)=\frac{1}{L}$ for
every $x\in\Omega$. Recall that the random variable $Y$ is obtained
via a random variable transformation from $X$ using the function
$s$. Hence, its density function exists on $\mathbb{R}-\mathcal{C}$---
where we have banished the image (under $s$) of the measure
zero set of points where $S^{\prime\prime}$ vanishes---and is given
by (\ref{eq:graddensity}). The reader may refer to \cite{Billingsley95}
for a detailed explanation. 
\end{proof}
\newpage
\subsection {Proof of Lemma~\ref{lemma:scaledPS}}
\label{sec:proofofscaledPSlemma}
\begin{proof}
By Parseval\rq{}s theorem we have
\begin{equation*}
 \frac{1}{N \delta} \sum_{k=0}^{N-1} \left|F^{D}(w_k)\right|^2 = \delta \sum_{n=0}^{N-1} \left|\phi_{\tau}^{D}(y_n)\right|^2 = \frac{N \delta}{L} = 1.
\end{equation*}
Noting that 
\begin{equation}
\label{eq:DFTscaledDFTrelation}
\sqrt{2 \pi \tau} F^D_{\tau}(u_k) = F^{D}(w_k), 
\end{equation}
the result follows immediately. 
\end{proof}

\subsection{Proof of No-stationary-points Lemma}
\label{sec:proofofnostationarypointslemma}
\begin{proof}
As $T^{\prime}(x;u)\not=0$, $T(x;u)$ is strictly monotonic. Integrating by parts we get
\begin{align}
W_{\tau}(u) =  &-i \tau \left[\frac{\Hdelta(x)\exp\left(\frac{i T(x;u)}{\tau}\right)}{s(x)-u}\right]_{b_1}^{b_2}\nonumber \\
\label{eq:Ibypartsintegral}
+&i \tau \int_{b_1}^{b_2}\exp\left(\frac{i T(x;u)}{\tau}\right)\left[q_1(x;u) - q_2(x;u)\right]\, dx
\end{align}
where
\begin{equation*}
q_1(x;u)=\frac{\Hdelta^{\prime}(x)}{s(x)-u} \hspace{10pt} \mbox{and} \hspace{10pt}
q_2(x;u)=\frac{\Hdelta(x)S^{\prime \prime}(x)}{\left[s(x)-u\right]^2}.
\end{equation*}
We split the integral in the right side of (\ref{eq:Ibypartsintegral}) into three parts by dividing at $x=0$ and $x=L$ where $\Hdelta^{\prime}$ is discontinuous. As $\Hdelta^{\prime}(x) = \frac{2}{\delta}$ and $s(x) = s(0)$ between $[b_1,0]$, $q_1(x;u) = \frac{2}{\delta (s(0)-u)}$ and $q_2(x,u) = 0$ as $S^{\prime \prime}(x) = 0$. Recalling that $\frac{\tau}{\delta} = \frac{C B}{\pi}$ we get
\begin{align}
\label{eq:Iintegral1complete}
\int_{b_1}^{0}\exp\left(\frac{i T(x;u)}{\tau}\right)q_1(x;u) dx = & \frac{-i 2CB}{\pi \left[s(0)-u\right]^2}\left[\exp\left(\frac{i T(0;u)}{\tau}\right) - \exp\left(\frac{i T(b_1;u)}{\tau}\right)\right]
\end{align}
On the portion $[L, b_2]$ where $\Hdelta^{\prime}(x) = \frac{-2}{\delta}$, $s(x) = s(L)$, and $q_2(x,u) = 0$ we have
\begin{align}
\label{eq:Iintegral3complete}
\int_{L}^{b_2}\exp\left(\frac{i T(x;u)}{\tau}\right)q_1(x;u)\, dx = \frac{i 2CB}{\pi  \left[s(L)-u\right]^2}\left[\exp\left(\frac{i T(b_2;u)}{\tau}\right) - \exp\left(\frac{i T(L;u)}{\tau}\right)\right]
\end{align}
We are left with the interval $[0, L]$ where $\Hdelta(x)$ being identically equal to $1$, $\Hdelta^{\prime}=0$ and $q_1(x;u)$ vanishes. Via integration by parts on the integral involving $q_2(x;u)$ we find
\begin{align}
\int_{0}^{L}\exp\left(\frac{i T(x;u)}{\tau}\right)q_2(u,x)\, dx =& -i \tau\left[\frac{S^{\prime \prime}(x)}{\left[s(x)-u\right]^3}\exp\left(\frac{i T(x;u)}{\tau}\right)\right]_{0}^L \nonumber \\
& +i\tau\int\limits_{0}^L \exp\left(\frac{i T(x;u)}{\tau}\right) \left[ \frac{S^{\prime \prime \prime}(x)}{\left[s(x)-u\right]^3} - \frac{3 \left[S^{\prime \prime}(x)\right]^2}{\left[s(x)-u\right]^4}\right] \, dx \nonumber \\
\label{eq:Iintegral2complete}
& = O\left(\frac{\tau}{\xi^3}\right)
\end{align}
where we have used the premise that $|s(x)-u| \geq \xi, \forall x \in [b_1,b_2]$.
Using the results (\ref{eq:Iintegral1complete}), (\ref{eq:Iintegral2complete}) and (\ref{eq:Iintegral3complete}) in (\ref{eq:Ibypartsintegral}) we get
\begin{align}
\label{eq:Ifinal1}
W_{\tau}(u) = & -i \tau \sum_{r=1}^2  (-1)^r \exp\left(\frac{i T(b_r;u)}{\tau}\right)\frac{\Hdelta(b_r)}{s(b_r)-u} \\
\label{eq:Ifinal2}
&+\frac{2CB\tau}{\pi \left[s(0)-u\right]^2}\left[\exp\left(\frac{i T(0;u)}{\tau}\right) - \exp\left(\frac{i T(b_1;u)}{\tau}\right)\right] \\
\label{eq:Ifinal3}
&+\frac{2CB\tau}{\pi  \left[s(L)-u\right]^2}\left[\exp\left(\frac{i T(L;u)}{\tau}\right) - \exp\left(\frac{i T(b_2;u)}{\tau}\right)\right] \\
\label{eq:Ifinal4}
&+O\left(\frac{\tau^2}{\xi^3}\right)
 \end{align}
which can be succinctly represented as $W_{\tau}(u) = O(\tau)$.
\end{proof}

\subsection{Proof of Lemma~\ref{lemma:productnostationarypoints}}
\label{sec:proofoflemma:productnostationarypoints}
\begin{proof}
Let $ p_{m}(u) = S(x_m(u))-S(\kappa)-u(x_m(u)-\kappa)-\gamma \kappa$. As $s(x_m(u)) = u$ we get $p_{m}^{\prime}(u) = \kappa - x_m(u) \not=0$ indicating that there are no stationary points. 
Defining
\begin{equation}
\label{eq:qm}
q_m(u) = \frac{1}{\left[s(\kappa)-(u-\gamma)\right]^2 \sqrt{\left|S^{\prime\prime}(x_m(u))\right|} (\kappa-x_m(u))}
\end{equation}
and integrating by parts we get
\begin{align*}
Z_{\tau} = &-i \tau \left[\exp\left(\frac{ip_m(u)}{\tau}\right) q_m(u) \right]_{\ut}^{\ut+\alpha_k} \\
&+i \tau \int\limits_{\Nalphat(\ut)}\exp\left(\frac{ip_m(u)}{\tau}\right) q_m^{\prime}(u) \, du.
\end{align*}
Knowing that $\frac{\,dx_m(u)}{\,du} = \frac{1}{S^{\prime\prime}(x_m(u))}$, $q_m^{\prime}(u)$ can be evaluated to be
\begin{align}
q_m^{\prime}(u) =& \frac{2}{\left[s(\kappa)-(u-\gamma)\right]^3 \left|S^{\prime\prime}(x_m(u))\right| [\kappa-x_m(u)]^2}  \nonumber \\
&+\frac{1}{\left[s(\kappa)-(u-\gamma)\right]^4 \left|S^{\prime\prime}(x_m(u))\right| (S^{\prime\prime}(x_m(u))) [\kappa-x_m(u)]^2} \nonumber \\
\label{eq:qmprime}
&-\frac{S^{\prime\prime \prime}(x_m(u))  \left|S^{\prime\prime}(x_m(u))\right|}{2\left[s(\kappa)-(u-\gamma)\right]^4 \left[\left|S^{\prime\prime}(x_m(u))\right|\right]^{\frac{3}{2}} (S^{\prime\prime}(x_m(u))) [\kappa-x_m(u)]^2}
\end{align}
We would like to emphasize the following inequality
\begin{equation}
\label{eq:kappainequality}
\kappa = \left\{ \begin{array}{ll}
              \rho_1; &|\kappa-x_m(u)| > |x_m(u)|; \\
              \rho_2; &|\kappa-x_m(u)| > |L-x_m(u)|.
          \end{array} \right.
\end{equation}
Furthermore, recall that $s(\rho_1) = s(0)$ and $s(\rho_2) = s(L)$ by construction and $\gamma = 2BCl, l \in \mathbb{Z}$ does not depend on $\delta$. Hence both $q_m(u)$ and $q_m^{\prime}(u)$ in (\ref{eq:qm}) and (\ref{eq:qmprime}) respectively can be individually bounded independent of $\delta$ (and also of $\tau$).
The result then follows.
\end{proof}

\subsection{Proof of Bound-on-Integrated-Error Lemma}
\label{sec:proofoflemma:boundsonintegrals}
\begin{proof}
Note that $\Nalphat(\ut)$ being a closed interval includes all the limit points. As $\{s(0),s(L)\} \in \C$ we have $\{s(0),s(L)\}\bigcap\Nalphat(\ut) = \emptyset$ by selection of $\Nalphat(\ut)$ as per Lemma~\ref{lemma:nostationarypoints}. Hence we could find $\xi_1,\xi_2 >0$ such that $|s(0)-u|\geq \xi_1$ and $|s(L)-u|\geq \xi_2, \forall u \in \Nalphat(\ut)$. As these distances are bounded away for zero, all the bounds obtained above for $u=\ut$ can be extended $\forall u \in \Nalphat(\ut)$.

Expressing the spatial locations $x_m$ as a function of $u$ using the inverse function $x_m(u) = s^{(-1)}(u)$, consider the phase of cross term defined in (\ref{eq:crossxmxt}), namely
\begin{equation*}
p_{m,t}(u) = S(x_m(u))-S(x_t(u))-u(x_m(u)-x_t(u)) +\theta_{m,t}(x_m(u),x_t(u))
\end{equation*}
where $x_m(u) \in \Nalpha(x_m)$, $x_t(u) \in \Nalpha(x_t)$ and $\Nalpha(x_m) \bigcap \Nalpha(x_t)  = \emptyset$ as $t\not=m$. Recall that  $\theta_{m,t}(x_m(u),x_t(u))$ depends on the sign of $S^{\prime \prime}(x(u))$ around $\Nalpha(x_m)$ and $\Nalpha(x_t)$. 
The constancy of the sign $S^{\prime \prime}(x_m(u))$ in $\Nalpha(x_m), \forall m$ by Lemma~\ref{lemma:neighborhoodLemma} removes the variability of $\theta_{m,t}(x_m(u),x_t(u))$ around the same region. Checking for the stationary condition while bearing in mind that $s(x_m(u)) = u$ we see that
\begin{equation*}
\left(p_{m,t}\right)^{\prime}(u) =  x_t(u) - x_m(u) \not=0,
\end{equation*}
signaling the absence of any stationary points. Integrating by parts once and proceeding along the proof of Lemma~\ref{lemma:productnostationarypoints} given in Appendix~\ref{sec:proofoflemma:productnostationarypoints} we get
\begin{equation}
\label{bound:integralct}
\int\limits_{\Nalphat(\ut)} \chi_{m,t,\tau}\left(x_m(u),x_t(u),u\right) \, du = O(\tau).
\end{equation}

Apropos to the bounds in (\ref{bound:aliasingerror}) and (\ref{bound:e3}) respectively, the magnitude square of both the aliasing error and $\epsilon_{3,\tau}(u)$ is $O(\tau), \forall u \in \Nalphat(\ut)$. The latter is also guaranteed by our single choice of $\lambda$ and $\epsilon, \forall u \in \Nalphat(\ut)$ as elucidated in Appendix~\ref{sec:e2bound} while deriving the bound for $\epsilon_{2,\tau}(\ut)$. By extending these bounds to the integral of these terms over $\Nalphat(\ut)$ we could reason that
\begin{align}
\label{bound:integrale3sq}
\int\limits_{\Nalphat(\ut)} \left|\epsilon_{3,\tau}(u)\right|^2 \, du &= O(\tau), \hspace{10pt} \mbox{and} \\
\label{bound:integralaesq}
\int\limits_{\Nalphat(\ut)}  \left| \sum_{l=-\infty, l \not=0}^{\infty}F_{\tau}(u-\gamma_l)\right|^2\, du &= O(\tau). \\
\end{align}

We leverage Lemma~\ref{lemma:productnostationarypoints} to bound the integral of $\epsilon_{4,\tau}(u)$ over $\Nalphat(\ut)$. 
Firstly, observe that the expressions in (\ref{eq:e3mainterm1}) and (\ref{eq:e3mainterm2}) are akin to the definition of $\zeta_{\tau}(u)$ in Lemma~\ref{lemma:productnostationarypoints}. Secondly, the remaining error term in (\ref{eq:e3error}) is $O(\tau),  \forall u \in \Nalphat(\ut)$. Furthermore, as the sign of $S^{\prime \prime}(x_m(u))$ is constant in $ \Nalphat(\ut)$, the $\pm \frac{\pi}{4}$ factor in the phase does not vary its sign. Applying Lemma~\ref{lemma:productnostationarypoints} we find
\begin{equation}
\label{bound:integrale4}
\int\limits_{\Nalphat(\ut)} \epsilon_{4,\tau}(u) \, du = O(\tau).
\end{equation}
We are left with computing
\begin{equation*}
\int\limits_{\Nalphat(\ut)}  F_{\tau}(u) \overline{\left(\sum_{l=-\infty, l \not=0}^{\infty} F_{\tau}(u-\gamma_l)\right)} \, du
\end{equation*}
and the integral of its conjugate.  Pursuant to the Lebesgue dominated convergence theorem we can switch the infinite summation and the integral allowing us to focus independently on $F_{\tau}(u-\gamma_l)$. Firstly, as $F_{\tau}(u)$ is a bounded function of $u$, the term in (\ref{eq:aeerrorbound}) when multiplied with $F_{\tau}(u)$ and integrated over $\Nalphat(\ut)$ produces a factor that is $O\left(\frac{\tau \sqrt{\tau}}{[B(|l|-1) +\beta]^3} \right)$. Secondly, recall that $ F_{\tau}(u)$ in (\ref{eq:Ftau_approx}) is composed of two terms where the error term $\epsilon_{3,\tau}(u) = O(\sqrt{\tau})$. The terms in  (\ref{eq:aesummation1}) and (\ref{eq:aesummation2}) being $O\left(\frac{\sqrt{\tau}}{\left[B(|l|-1) +\beta \right]^2}\right), \forall u \in \Nalphat(\ut)$ when multiplied with $\overline{\epsilon_{3,\tau}(u)}$ and integrated results in an expression that is $O\left(\frac{\tau}{\left[B(|l|-1) +\beta \right]^2}\right)$. To bound the integration of the product of first (main) term on the right of (\ref{eq:Ftau_approx}) with the expressions in  (\ref{eq:aesummation1}) and (\ref{eq:aesummation2}), we employ Lemma~\ref{lemma:productnostationarypoints} and find it to be $O\left(\frac{\tau \sqrt{\tau}}{\left[B(|l|-1) +\beta \right]^2}\right)$. Coupling these results we have
\begin{equation*}
\int\limits_{\Nalphat(\ut)}  F_{\tau}(u)  \overline{F_{\tau}(u-\gamma_l)} \, du = O\left(\frac{\tau}{\left[B(|l|-1) +\beta \right]^2}\right).
\end{equation*}
The infinite summation then leads to
\begin{align}
\int\limits_{\Nalphat(\ut)}  F_{\tau}(u) \overline{\left(\sum_{l=-\infty, l \not=0}^{\infty} F_{\tau}(u-\gamma_l)\right)} \, du & = \sum_{l=-\infty, l \not=0}^{\infty} \int\limits_{\Nalphat(\ut)}  F_{\tau}(u) \overline{F_{\tau}(u-\gamma_l)} \, du \nonumber \\
\label{bound:integralFwithae}
& = O(\tau).
\end{align}
Combining (\ref{bound:integralct}), (\ref{bound:integrale3sq}), (\ref{bound:integralaesq}), (\ref{bound:integrale4}), and (\ref{bound:integralFwithae}), the proof follows.
\end{proof}

\section{Expression for the error $\mathbf{\epsilon_{2,\tau}}$}
\label{sec:e2bound}
Let 
\begin{equation*}
\epsilon_{2,\tau}(\ut) = \sum_{t=1}^{M(\ut)}\epsilon_{2,t,\tau}(\ut) + \tilde{\epsilon}_{2,t,\tau}(\ut) 
\end{equation*}
where $\epsilon_{2,t,\tau}(\ut)$ and $\tilde{\epsilon}_{2,t,\tau}(\ut)$ are the stationary phase errors incurred while evaluating $K_{t,\tau}(\ut)$ and $\tilde{K}_{t,\tau}(\ut)$ respectively. As before, let the finite set $\{x_t\}_{t=1}^{M(\ut)}$ be the location of the stationary points for the given $\ut$. The Theorem~13.1 in Chapter~3 of \cite{OlverBook74} expresses $\epsilon_{2,t,\tau}(\ut)$ as
\begin{equation}
\label{eq:e2terror}
\epsilon_{2,t,\tau}(\ut) =  -\epsilon_{2,1,t,\tau}(\ut) + \epsilon_{2,2,t,\tau}(\ut)
\end{equation}
where
\begin{align}
\label{eq:e21error}
\epsilon_{2,1,t,\tau}(\ut)  &= \frac{\exp\left(\frac{i T(x_t;\ut)}{\tau}\right)}{\sqrt{2\left|S^{\prime \prime}(x_t)\right|}} \int\limits_{\eta}^{\infty}\exp\left(\frac{i v}{\tau}\right) v^{-1/2} \, dv  \hspace{10pt} \mbox{and} \\
\label{eq:e22error}
\epsilon_{2,2,t,\tau}(\ut) &= \exp\left(\frac{i T(x_t;\ut)}{\tau}\right) \int\limits_{0}^{\eta} \exp\left(\frac{i v}{\tau}\right) \Q(v;\ut) \, dv.
\end{align}
Here
\begin{align}
v &= T(x;\ut)- T(x_t;\ut), \nonumber \\
\eta &= T(c_t;\ut)- T(x_t;\ut)  \hspace{10pt} \mbox{and} \nonumber \\
\label{def:Q}
\Q(v;\ut) = &\frac{1}{s(x(v))-\ut} - \frac{1}{\sqrt{2 |S^{\prime \prime}(x_t)|}v^{\frac{1}{2}}}.
\end{align}
As $T(x;\ut)$ is strictly monotonic in $[x_t,c_t]$ it is proper to express $x$ as a function of $v$. Evaluating (\ref{eq:e21error}) by integration by parts twice we get
\begin{equation}
\label{eq:e21errorexpressed}
\epsilon_{2,1,t,\tau}(\ut) = i \tau \frac{\exp\left(\frac{i T(c_t;\ut)}{\tau}\right)}{\sqrt{2\left|S^{\prime \prime}(x_t)\right|} \left[T(c_t;\ut)-T(x_t;\ut)\right]^{\frac{1}{2}}} + O(\tau^2).
\end{equation}

We obtain the expression for $\epsilon_{2,2,t,\tau}(\ut)$ by pursuing along the lines of Theorem~12.3 and Theorem~13.1 in Chapter~3 of the book \cite{OlverBook74} and the article \cite{OlverArticle74}. Prior to delving into the details we would like to underscore that the $O(\tau \sqrt{\tau})$ error bound in (\ref{bound:intQprime}) deduced for our specialized stationary phase approximation setting (for e.g. assuming that $S^{\prime \prime}(x_t)$ does not vanish) is \emph{stronger} than the $O(\tau)$ bound derived in the aforementioned citations where the author studies a broader scenario. This stronger result is key to our $O(1/N)$ approximation result. 

As shown in \cite{OlverArticle74}, for small values of $v$, $x-x_t$ can be expanded---as a function of $v$---in asymptotic series of the form 
\begin{equation}
\label{eq:xseriesexpansion}
x-x_t \sim \sum\limits_{l=0}^\infty d_l v^{(l+1)/2}
\end{equation}
where the coefficients $d_l$ may be obtained by following the standard procedures for reverting the series. In particular $d_0  = \sqrt{\frac{2}{|S^{\prime \prime}(x_t)|}}$. The other constants $d_1,d_2,\cdots$ are a function of \emph{second and higher derivatives} of $T$ around the stationary point $x_t$. Hence they depend only on the nature of the function $S$ around $x_t$ and not directly on the frequency $\ut$. The indirect dependency on $\ut$ is only through its corresponding stationary point $x_t$ as elucidated below. Differentiating with respect to $v$ we get
\begin{equation*}
\frac{1}{s(x(v))-\ut} = \frac{\,dx}{\,dv} \sim \sum\limits_{l=0}^\infty \left(\frac{l+1}{2}\right) d_l v^{(l-1)/2}.
\end{equation*}
Letting $a_l = \frac{l+1}{2} d_l$, $\Q(v;\ut)$ can be seen to admit a series expansion,
 \begin{equation}
 \label{eq:Qseriesexpansion}
 \Q(v;\ut) \sim a_1+a_2 v^{\frac{1}{2}} + a_3 v + a_4 v^{\frac{3}{2}} + \cdots,
 \end{equation}
as $v \rightarrow 0^{+}$.  It is also shown in \cite{OlverArticle74} that
\begin{align}
\label{bound:finiteintegral}
\int\limits_{0}^{\eta} \left|\Q^{\prime}(v;\ut)\right| \, dv < \infty.
\end{align}
 Computing (\ref{eq:e22error}) by integration by parts and noticing that $\lim\limits_{v \rightarrow 0^+} \Q(v;\ut) = a_1$ we get
\begin{align}
\label{eq:e22errorexpressed}
\epsilon_{2,2,t,\tau}(\ut) = &-i\tau \exp\left(\frac{i T(c_t;\ut)}{\tau}\right)\left[\frac{1}{s(c_t)-\ut} - \frac{1}{\sqrt{2\left|S^{\prime \prime}(x_t)\right|} \left[T(c_t;\ut)-T(x_t;\ut)\right]^{\frac{1}{2}}}\right] \\
& + a_1  i \tau \exp\left(\frac{i T(x_t;\ut)}{\tau}\right) \\
\label{eq:intQprime}
&+i \tau \exp\left(\frac{i T(x_t;\ut)}{\tau}\right) \int\limits_{0}^{\eta}  \exp\left(\frac{i v}{\tau}\right) \Q^{\prime}(v;\ut) \, dv.
\end{align}
The finiteness of (\ref{bound:finiteintegral}) assures that (\ref{eq:intQprime}) is bounded. Our next task is to capture this bound as a function of $\tau$.

Based on the series expression for $ \Q(v;\ut)$ in (\ref{eq:Qseriesexpansion}) we see that $\Q^{\prime}(v;\ut) = O\left(v^{\frac{-1}{2}}\right)$ and $\Q^{\prime \prime}(v;\ut) = O\left(v^{\frac{-3}{2}}\right)$ \emph{independent} of $\tau$ as $v \rightarrow 0^{+}$. Then there exist constants $\lambda>0$ and $\epsilon>0$---independent of $\tau$---such that $\left|\Q^{\prime}(v;\ut)\right| \leq \epsilon v^{\frac{-1}{2}}$ and $\left|\Q^{\prime \prime}(v;\ut)\right| \leq \epsilon v^{\frac{-3}{2}}$ when $v \leq \lambda$. In the subsequent paragraph we would like to add an important technical note on the choice of $\lambda$ and $\epsilon$. The reader may choose to skip the next paragraph without loss of continuity but bear in mind to refer to it when we discuss the proof of Lemma~\ref{lemma:boundsonintegrals} in Appendix~\ref{sec:proofoflemma:boundsonintegrals}.

As mentioned above, the constants $d_1, d_2, \cdots$ in (\ref{eq:xseriesexpansion}) depend only on the property of $S$ around $x_t$ and not directly on $\ut$. However, as $\ut$ is varied (say) over a small compact interval $\Nalphat(\ut)$ (which we soon require in Lemma~\ref{lemma:boundsonintegrals}), the corresponding stationary point $x_t(\ut)$, now explicitly expressed as the function of $\ut$, moves around in the compact interval $s^{(-1)}\left(\Nalphat(\ut)\right) = \Nalphat(x_t)$ influencing the constants in (\ref{eq:xseriesexpansion}) and creating an indirect dependency of them on $\ut$. It can be verified from \cite{Olver70} that the constant $d_2$ (and thereby $a_2$) which decides the aforementioned growth rate of $\Q^{\prime}(v;\ut)$ and $\Q^{\prime \prime}(v;\ut)$ as $v \rightarrow 0^{+}$ varies $\propto \frac{1}{\left[S^{\prime \prime}(x_t(\ut))\right]^{2+3/2}}$ with $S^{\prime \prime}(x_t(\ut))$ being the only derivative of $S$ appearing in the denominator. As we proceed, we will soon see that our choice of neighborhood $\Nalphat(\ut)$ will be pursuant to Lemma~\ref{lemma:neighborhoodLemma} where in $\Nalphat(x_t)$, $S^{\prime \prime}(x_t(\ut)) \not=0$ and bounded away from zero. This in turn enable us to choose a single value for each the constants $\lambda$ and $\epsilon$ for all $\ut \in \Nalphat(\ut)$. 

Since we are interested in $N \rightarrow \infty$ or equivalently $\tau \rightarrow 0$, let $\tau$ be such that $\tau \leq \lambda$. Our subsequent steps closely follow Theorem~12.3 in Chapter~3 of \cite{OlverBook74}. Lack of a strong constraint---$\Q^{\prime}(v;\ut) = o\left(v^{\frac{-1}{2}}\right)$ and $\Q^{\prime \prime}(v;\ut) = o\left(v^{\frac{-3}{2}}\right)$---precludes us from directly applying Theorem~12.3 to prove a stronger assertion. However, the weaker constraints on $\Q^{\prime}$ and $\Q^{\prime \prime}$ ($O$ instead of $o$) leads to an equivalently weak but a sufficiently strong result.

\noindent Dividing the integral (\ref{eq:intQprime}) at $v = \tau$ we get
\begin{align}
\label{bound:intQprime1}
\left|\int\limits_{0}^{\tau}  \exp\left(\frac{i v}{\tau}\right) \Q^{\prime}(v;\ut) \, dv \right| \leq \epsilon \int\limits_{0}^{\tau}  v^{\frac{-1}{2}} \, dv = 2 \epsilon \sqrt{\tau}.
\end{align}
Using integration by parts we find
\begin{align}
\int\limits_{\tau}^{\eta}  \exp\left(\frac{i v}{\tau}\right) \Q^{\prime}(v;\ut) \, dv  = & \frac{\tau}{i}\left[\exp\left(\frac{i \eta}{\tau}\right) \Q^{\prime}(\eta;\ut) -  \exp\left(\frac{i \tau}{\tau}\right) \Q^{\prime}(\tau;\ut)\right] \nonumber \\
&- \frac{\tau}{i} \int\limits_{\tau}^{\lambda}  \exp\left(\frac{i v}{\tau}\right) \Q^{\prime \prime}(v;\ut) \, dv \nonumber \\
\label{eq:intQprime2ndpartexpressed}
&- \frac{\tau}{i} \int\limits_{\lambda}^{\eta}  \exp\left(\frac{i v}{\tau}\right) \Q^{\prime \prime}(v;\ut) \, dv.
\end{align}
Recalling that $\Q^{\prime \prime}(v;\ut) \leq \epsilon \tau^{\frac{-3}{2}}$ when $v \leq \lambda$ we further have
\begin{equation}
\label{bound:intQprime2}
\left|- \frac{\tau}{i} \int\limits_{\tau}^{\lambda}  \exp\left(\frac{i v}{\tau}\right) \Q^{\prime \prime}(v;\ut) \, dv \right| \leq \tau \int\limits_{\tau}^{\lambda} \epsilon v^{\frac{-3}{2}} \,dv  =  2 \tau \epsilon \left(\frac{1}{\sqrt{\tau}} - \frac{1}{\sqrt{\lambda}}\right) < 2 \epsilon \sqrt{\tau},
\end{equation}
and as $\lambda$ is independent of $\tau$ and $\Q^{\prime \prime}(v;\ut)$ is bounded away from zero for $v \in [\lambda,\eta]$ we get
\begin{equation}
\label{bound:intQprime3}
\left|\frac{\tau}{i} \int\limits_{\lambda}^{\eta}  \exp\left(\frac{i v}{\tau}\right) \Q^{\prime \prime}(v;\ut) \, dv \right| \leq \tau \int\limits_{\lambda}^{\eta} \left| \Q^{\prime \prime}(v;\ut) \right| \, dv = O(\tau).
\end{equation}
Using the bound $\left|\Q^{\prime}(\tau;\ut)\right| \leq \epsilon \tau^{\frac{-1}{2}}$ in (\ref{eq:intQprime2ndpartexpressed}) and combining (\ref{bound:intQprime1}), (\ref{eq:intQprime2ndpartexpressed}), (\ref{bound:intQprime2}) and (\ref{bound:intQprime3}) we arrive at
\begin{equation}
\label{bound:intQprime}
i \tau \int\limits_{0}^{\eta}  \exp\left(\frac{i v}{\tau}\right) \Q^{\prime}(v;\ut) \, dv = O(\tau \sqrt{\tau})
\end{equation}
as $\tau \rightarrow 0$ ($N \rightarrow \infty$). Plugging (\ref{bound:intQprime}) in (\ref{eq:e22errorexpressed}) and subtracting (\ref{eq:e21errorexpressed}) gives us
\begin{align}
\epsilon_{2,t,\tau}(\ut) = & i \tau a_1  \exp\left(\frac{i T(x_t;\ut)}{\tau}\right) -i\tau \frac{\exp\left(\frac{i T(c_t;\ut)}{\tau}\right)}{s(c_t)-\ut} \nonumber \\
\label{eq:e2terrorexpressed}
& + O(\tau^2) + O(\tau \sqrt{\tau}).
\end{align}
We would like to add the following important remark about the first term on the right side of (\ref{eq:e2terrorexpressed}). The computation of the error $\tilde{\epsilon}_{2,t,\tau}(\ut)$ along similar lines on the interval $[c_{t-1},x_t]$ will produce the exact expression but with a \emph{negative sign}. These two terms cancel with each other leaving no expression in $\epsilon_{2,\tau}(u)$ containing $T(x_t;\ut)$ in the phase. The total stationary phase error at the critical point $x_t$ equals
\begin{align*}
\epsilon_{2,t,\tau}(\ut)  + \tilde{\epsilon}_{2,t,\tau}(\ut)  = & i\tau \frac{\exp\left(\frac{i T(c_{t-1};\ut)}{\tau}\right)}{s(c_{t-1})-\ut} - i\tau \frac{\exp\left(\frac{i T(c_t;\ut)}{\tau}\right)}{s(c_t)-\ut} \\
& + O(\tau \sqrt{\tau}).
\end{align*}
Being a telescopic series the adjacent terms cancel each other when summed and
\begin{align}
\label{eq:e2errorexpressed}
\epsilon_{2,\tau}(u) = & i\tau \frac{\exp\left(\frac{i T(c_0;\ut)}{\tau}\right)}{s(c_0)-\ut} - i\tau \frac{\exp\left(\frac{i T(c_{M(\ut)};\ut)}{\tau}\right)}{s(c_{M(\ut)})-\ut} \\
\label{eq:e2errorterm}
&+O(\tau \sqrt{\tau}).
\end{align}

\bibliographystyle{siamplain}
\bibliography{ConvergenceRateGradDensityEstimation}
\end{document}